\documentclass[11pt,envcountsame,envcountsect]{llncs}
\usepackage{geometry}
\usepackage{amsmath,amssymb}
\usepackage{paralist}
\usepackage{graphicx}
\usepackage{xspace}
\usepackage[utf8x]{inputenc}
\usepackage[unicode]{hyperref}

\hypersetup{
  pdftitle={Online Bin Stretching with Three Bins},
  pdfauthor={Martin B\"{o}hm, Ji\v{r}\'{\i} Sgall, Rob van Stee, Pavel Vesel\'y},
}

\spnewtheorem{observation}[theorem]{Observation}{\bfseries}{\itshape}
\spnewtheorem{goodsit}{Good Situation}{\bfseries}{\itshape}
\spnewtheorem{gengoodsit}{Good Situation}{\bfseries}{\itshape} % to avoid pdftex warning about duplicate labels
\spnewtheorem{clm}[theorem]{Claim}{\bfseries}{\itshape}

\newcommand\binstretch{{\sc Online Bin Stretching}\xspace}
\newcommand\binpacking{{\sc Bin Packing}\xspace}
\newcommand\Nat{\mathbb{N}}

\newcommand\eps\varepsilon
\newcommand\calA{{\mathcal A}}

\newcommand\calC{{\mathcal C}}

\newcommand\calL{{\mathcal L}}

\newcommand{\before}{\leftarrow}

\newcommand*\samethanks[1][\value{footnote}]{\footnotemark[#1]}

\def\O{{\mathcal{O}}}
\let\csc=\textsc
\def\GSFF#1{{\rm\csc{GSFF(#1)}}}
\def\FF{{\rm\csc{First Fit}}\xspace}
\def\tbalg{\textsc{Evasive}\xspace}

\def\indentskip{\hskip 1.5em}
\def\indentskiptwodigit{\hskip 1em}

\def\step#1{{(\ref{fp:#1})}}
\def\eq#1{{(\ref{eq:#1})}}
\def\gs#1{GS\ref{lem:gs#1}}

\long\def\algobox#1{\smallskip
  \noindent
~~\hbox{\fbox{\parbox[c]{0.95\textwidth}{#1}}}
\smallskip
}

\def\frombox#1{\hbox to 1em{{\bf #1}\hss}}
\def\tobox#1{\hbox to 2.3em{{\bf #1}\hss}}
\def\arrow{\hbox to 1.5em{$\rightarrow$}}

\makeatletter
\newcommand*{\defeq}{\mathrel{\rlap{%
                     \raisebox{0.3ex}{$\m@th\cdot$}}%
                     \raisebox{-0.3ex}{$\m@th\cdot$}}%
                     =}
\makeatother

\newcommand\It{{\mathcal I}}
\renewcommand\Game{{\mathrm{BSG}}}
\newcommand\algo{\textsc{Algorithm}\xspace}
\newcommand\adversary{\textsc{Adversary}\xspace}
\newcommand\evaladv{\textsc{EvaluateAdversary}\xspace}
\newcommand\evalalg{\textsc{EvaluateAlgorithm}\xspace}
\newcommand\Main{\textsc{Main}\xspace}
\newcommand\Test{\textsc{Test}\xspace}
\newcommand\bfd{\textsc{Best Fit Decreasing}\xspace}

\newlength{\elseskiplen}
\newcommand\elseskip{\hskip\elseskip}
\begin{document}

\title{Online Bin Stretching with Three Bins}

\author{Martin B\"{o}hm\inst{1}\fnmsep\thanks{
Supported by the project
14-10003S of GA \v{C}R and by the GAUK project 548214.
},
Ji\v{r}\'{\i} Sgall\inst{1}\fnmsep\samethanks,
Rob van Stee\inst{2}
\and
Pavel Vesel\'y\inst{1}\fnmsep\samethanks
}

\institute{
Computer Science Institute of Charles University,
%Faculty of Mathematics and Physics,
%Malostransk\'e n\'am.\ 25, CZ-11800 Praha 1,
Prague, Czech Republic.
\email{\{bohm,sgall,vesely\}@iuuk.mff.cuni.cz}.
\and
Department of Computer Science,
University of Leicester, Leicester, UK.
\email{rob.vanstee@leicester.ac.uk}. 
}

\maketitle
% remove if submitting into LNCS
\thispagestyle{plain}
\begin{abstract}

\binstretch is a semi-online variant of bin packing in which the
algorithm has to use the same number of bins as an optimal packing,
but is allowed to slightly overpack the bins. The goal is to minimize
the amount of overpacking, i.e., the maximum size packed into any bin.

We give an algorithm for \binstretch with a stretching factor of $11/8 = 1.375$
for three bins. Additionally, we present a
lower bound of $45/33 = 1.\overline{36}$ for \binstretch on three bins 
and a lower bound of $19/14$ for four and five bins that were discovered using a
computer search.

\end{abstract}

\section{Introduction}

The most famous algorithmic problem dealing with online assignment is
arguably {\sc Online Bin Packing}. In this problem, known since the
1970s, items of size between $0$ and $1$ arrive in a
sequence and the goal is to pack these items into the least number of
unit-sized bins, packing each item as soon as it arrives.

\binstretch, which was introduced by Azar and Regev in 1998~\cite{azar98,azar01}, deals with a
similar online scenario. Again, items of size between $0$ and $1$
arrive in a sequence, and the algorithm needs to pack them as soon as
each item arrives, but it has two advantages: (i) The packing
algorithm knows $m$, the number of bins that an optimal offline
algorithm would use, and must also use only at most $m$ bins, and (ii)
the packing algorithm can use bins of capacity $R$ for some
$R\geq 1$. The goal is to minimize the stretching factor $R$.

While formulated as a bin packing variant, \binstretch can also be
thought of as a semi-online scheduling problem, in which we schedule
jobs in an online manner on exactly $m$ machines, before any execution
starts.  We have a guarantee that the optimum offline algorithm could
schedule all jobs with makespan $1$. Our task is to present an online
algorithm with makespan of the schedule being at most $R$.

% \iffalse % motivation can be skipped to save space
\smallskip\noindent
{\bf Motivation.} We give two of applications of \binstretch.

{\it Server upgrade.} This application has first appeared
in~\cite{azar98}. In this setting, an older server (or a server
cluster) is streaming a large number of files to the newer server
without any guarantee on file order. The files cannot be split between
drives. Both servers have $m$ disk
drives, but the newer server has a larger capacity of each drive. The
goal is to present an algorithm that stores all incoming files from
the old server as they arrive.

{\it Shipment checking.} A number $m$ of containers arrive at a
shipping center. It is noted that all containers are at most $p \le 100$ percent
full. The items in the containers are too numerous to be individually
labeled, yet all items must be unpacked and scanned for illicit and dangerous
material. After the scanning, the items must be speedily
repackaged into the containers for further shipping.
In this scenario, an algorithm with stretching factor $100/p$
can be used to repack the objects into containers in an online manner.
% \fi

\smallskip
\noindent 
{\bf History.}  
\binstretch was proposed by Azar and Regev~\cite{azar98,azar01}.
Already before this, a matching upper and lower bound of $4/3$ for two bins
had appeared~\cite{KeKoST97}.
%The original lower bound of $4/3$ for three bins has appeared even
%before that, in~\cite{KeKoST97}, for two bins together with a matching
%algorithm.  
Azar and Regev extended this lower bound to any number
of bins and gave an online algorithm with a stretching factor $1.625$.

The problem has been revisited recently, with both lower bound
improvements and new efficient algorithms. On the algorithmic side,
Kellerer and Kotov~\cite{kellerer2013} have achieved a stretching
factor $11/7 \approx 1.57$ and Gabay et al.~\cite{gabay2013} 
have achieved $26/17 \approx
1.53$. The best known general algorithm with stretching factor $1.5$
was presented by the authors of this paper in \cite{bsvsv16}.

In the case with only three bins, the previously best
algorithm was due to Azar and Regev~\cite{azar98}, with a stretching
factor of $1.4$.

On the lower bound side, the lower bound $4/3$ of~\cite{azar98} was
surpassed only for the case of three bins by Gabay et
al.~\cite{gabay2013lbv2}, who show a lower bound of $19/14\approx
1.357$, using an extensive computer search. The preprint
\cite{gabay2013lbv2} was updated in 2015 \cite{gabay2013lbv3} to
include a lower bound of $19/14$ for four bins.

\smallskip\noindent
{\bf Our contributions.} In Section~\ref{sec:three118} we present an
algorithm for three bins of capacity $11/8 = 1.375$.  This is the
first improvement of the stretching factor $1.4$ of Azar and
Regev~\cite{azar98}.
In Section~\ref{sec:lowerbound}, we present a new lower bound of
$45/33 = 1.\overline{36}$ for \binstretch on three bins, along with a
lower bound of $19/14$ on four and five bins which is the first
non-trivial lower bound for four and five bins. We build on the
paper of Gabay et al.~\cite{gabay2013lbv2} but significantly change the
implementation, both technically and conceptually. The lower bound of $19/14$
for four bins is independently shown in \cite{gabay2013lbv3}.

A preliminary version of this work appeared in WAOA 2014
\cite{bsvsv14} and SOFSEM 2016~\cite{bohm16}.

\smallskip\noindent
{\bf Related work.} The NP-hard problem \binpacking was originally
proposed by Ullman~\cite{ullman71} and Johnson~\cite{johnson73} in the
1970s. Since then it has seen major interest and progress, see the
survey of Coffman et al.~\cite{coffman13} for many results on
classical Bin Packing and its variants. While our problem can be seen
as a variant of {\sc Bin Packing}, note that the algorithms cannot
open more bins than the optimum and thus general results for
\binpacking do not translate to our setting.

As noted, \binstretch can be formulated as the online scheduling on $m$ identical
machines with known optimal makespan. Such algorithms were studied and
are important in designing constant-competitive algorithms without the
additional knowledge, e.g., for scheduling in the more general model
of uniformly related machines~\cite{AsAFPW97,BeChKa00,EbJaSg09}.

For scheduling, also other types of semi-online algorithms are
studied. Historically first is the study of ordered sequences with
non-decreasing processing times~\cite{Graham69}. Most closely
related is the variant with known sum of all processing times studied
in~\cite{KeKoST97} and the currently best results are a lower bound of
$1.585$ and an algorithm with ratio $1.6$, both
from~\cite{DBLP:journals/tcs/AlbersH12}. Note that this shows,
somewhat surprisingly, that knowing the actual optimum gives a
significantly bigger advantage to the online algorithm over knowing
just the sum of the processing times (which, divided by $m$, is a
lower bound on the optimum).

\smallskip\noindent
{\bf Definitions and notation.} Our main problem, \binstretch, can be
described as follows: 

\noindent
{\bf Input:} an integer $m$ and a sequence
of items $I=i_1, i_2, \ldots$ given online one by one. Each item has
a {\it size} $s(i) \in [0,1]$ and must be packed immediately
and irrevocably. It is guaranteed that 
there exists a packing of all items in $I$ into $m$ bins of capacity $1$.

\noindent
{\bf Parameter:} The {\it stretching factor} $R\geq 1$.

\noindent
{\bf Output:} Partitioning (packing) of $I$ into bins $B_1,\ldots,B_m$
so that $\sum_{i\in B_j}s(i)\leq R$ for all $j=1,\ldots,m$.

%\noindent
%{\bf Guarantee:} there exists a packing of all items in $I$ into $m$ bins of capacity $1$.

\noindent
{\bf Goal:} Design an online algorithm with the stretching factor $R$
as small as possible which packs all input sequences satisfying the
guarantee.
\smallskip

For a bin $B$, we define the {\it size of the bin} $s(B) = \sum_{i \in
  B} s(i)$.  Unlike $s(i)$, $s(B)$ can change during the course of the
algorithm, as we pack more and more items into the bin. To easily
differentiate between items, bins and lists of bins, we use lowercase
letters for items ($i$, $b$, $x$), uppercase letters for bins and
other sets of items ($A$, $B$, $X$), and calligraphic letters for
lists of bins ($\calA$, $\calC$, $\calL$).

%In both sections of our paper, w
We rescale the item sizes and bin
capacities for simplicity. Therefore, in our setting, each item
has an associated size $s(i) \in [0,k]$, where $k \in \Nat$ is also the
capacity of the bins which the optimal offline algorithm uses. The
online algorithm for \binstretch uses bins of capacity
$t\in\Nat$, $t\geq k$. The resulting stretching factor is thus
$t/k$.

\section{Algorithm}\label{sec:three118}

%In this section, we focus on the problem of bin stretching restricted
%to only three bins. 

We scale the input sizes by $16$. The stretched bins in our setting
therefore have capacity $22$ and the optimal offline algorithm can
pack all items into three bins of capacity $16$ each.
We prove the following theorem.

\begin{theorem}\label{thm:3binsalgo}
There exists an algorithm that solves
\binstretch for three bins with stretching factor $1+3/8=1.375$.
\end{theorem}

The three bins of our setting are named $A$, $B$, and $C$.
We exchange the names of bins sometimes during the course of the algorithm.

A natural idea is to try to pack first all items in a single bin, as
long as possible. In general, this is the strategy that we
follow. However, somewhat surprisingly, it turns out that from the
very beginning we need to put items in two bins even if the items as
well as their total size are relatively small. 

It is clear that we have to be very cautious about exceeding a load of
6. For instance, if we put 7 items of size 1 in bin $A$, and 7 such items in $B$,
then if two items of size 16 arrive, the algorithm will have a load of 
at least 23 in some bin. Similarly, we cannot assign too much to a single bin:
putting 20 items of size 0.5 all in bin $A$ gives a load of 22.5 somewhere if 
three items of size 12.5 arrive next. (Starting with items of size 0.5
guarantees that there is a solution with bins of size 16 at the end.) 

On the other hand, it is useful to keep one bin empty for some time; many
problematic instances end with three large items such that one of them has
to be placed in a bin that already has high load. Keeping one bin free ensures
that such items must have size more than 11 (on average), which limits the
adversary's options, since all items must still fit into bins of size 16.

Deciding when exactly to start using the third bin and when to cross the
threshold of 6 for the first time was the biggest challenge in designing
this algorithm: both of these events should preferably be postponed as long
as possible, but obviously they come into conflict at some point.

\subsection{Good situations}\label{sec:gs}
Before stating the algorithm itself, we list a number of {\it good
  situations} (GS). These are configurations of the three bins which
allow us to complete the packing regardless of the following input.

It is clear that the identities of the bins are not important here;
for instance, in the first good situation, all that we need is that
\emph{any} two bins together have items of size at least 26. We have
used names only for clarity of presentation and of the proofs.

\begin{definition}
A \textit{partial packing} of an input sequence $S$
is a function $p: S_1 \to \{A,B,C\}$ that assigns a bin to each
item from a prefix $S_1$ of the input sequence $S$.
\end{definition}

\begin{goodsit}\label{lem:gs1}
Given a partial packing such that $s(A) + s(B) \geq 26$ and $s(C)$ is
arbitrary,
there exists an online algorithm that packs all remaining items into
three bins of capacity $22$.
\end{goodsit}

\begin{proof}
Since the optimum can pack into three bins of size $16$, the total
size of items in the instance is at most $3\cdot 16=48$. If two bins
have size $s(A) + s(B) \geq 26$, all the remaining items (including
the ones already placed on $C$) have size at most $22$.  Thus we can
pack them all into bin $C$.
\qed
\end{proof}

\begin{goodsit}\label{lem:gs2}
Given a partial packing such that $s(A) \in [4, 6]$ and $s(B)$ and
$s(C)$ are arbitrary, there exists
an online algorithm that packs all remaining items into three bins
of capacity $22$.
\end{goodsit}

\begin{proof}
Let $A$ be the bin with size between $4$ and $6$ and $B$ be one of the
other bins (choose arbitrarily). Put all the items greedily into
$B$. When an item does not fit, put it into $A$, where it fits, as
originally $s(A)$ is at most $6$. Now the size of all items in $B$
plus the last item is at least $22$. In addition, $A$ has items of
size at least $4$ before the last item by the assumption.  Together we
have $s(A) + s(B) \geq 26$, allowing us to apply \gs1.
\qed
\end{proof}

From now on, we assume that each bin $X\in\{A,B,C\}$ satisfies $s(X) \not\in [4,6]$, otherwise we reach \gs2.

\begin{goodsit}\label{lem:gs3}
Given a partial packing such that $s(A) \ge 15$ and either
{\rm(i)} $s(B)+s(C)\ge 22$ or
{\rm(ii)} $s(C)<4$ and $s(B)$ is arbitrary,
there exists an online algorithm that packs all remaining items into
three bins of capacity $22$. 
\end{goodsit}

\begin{proof}
(i) We have
%If $s(B)+s(C)\ge22$, then 
$\max(s(B),s(C))\ge11$, so 
we are in \gs1 on bins
$A$ and $B$ or on bins $A$ and $C$. %Else, if $s(C)<4$, w

(ii) We pack arriving items into $B$. If $s(B) \geq
11$ at any time, we apply \gs1 on bins $A$ and $B$. Thus we can
assume $s(B) < 11$ and we cannot continue packing into $B$ any
further. This implies that an item $i$ arrives such that $s(i) >
11$. As $s(C)<4$, we pack $i$ into it and apply \gs1 on
bins $A$ and $C$.\qed
\end{proof}

\begin{goodsit}\label{lem:gs4}
Given a partial packing such that 
$s(A) +s(B) \geq 15+\frac{1}{2}s(C)$, $s(B) < 4$, and $s(C) < 4$,
there exists an online algorithm that packs all remaining items into
three bins of capacity $22$.
\end{goodsit}

\begin{proof}
Let $c$ be the value of $s(C)$ when the conditions of this good situation
hold for the first time.
We run the following algorithm until we reach \gs1 or \gs3:

\smallskip
\algobox{
\begin{compactenum}[(1)]
\item If the incoming item $i$ has $s(i)\ge 11-\frac{1}{2}c$, pack $i$ into $B$.
\item Else, if $i$ fits on $A$, pack it there.
\item Otherwise pack $i$ into $C$.
\end{compactenum}
}
\medskip

If at any time an item $i$ is packed into $B$ (where it always fits),
then $s(A)+s(B)\ge 26$ and we reach \gs1. In the event that no item $i$
is packed into $B$, we reach \gs3 (with $B$ in the role of $C$) whenever the algorithm brings the size
of $A$ to or above 15.

The only remaining case is when $s(A)<15$ throughout the algorithm and
several items with size in the interval $I \defeq
(22-s(A),11-\frac{1}{2}c)$ arrive. These items are packed into $C$.
Note that $I \subseteq (7,11)$ and that the lower bound of $I$ may decrease
during the course of the algorithm.

The first two items with size in $I$ will fit together,
since $2(11-\frac{1}{2}c) + c = 22$. With two such items packed into $C$, we know that the load $s(A) + s(C)$ is at least $s(A) + 2(22-s(A)) = 44 - s(A) > 29$ and we have reached \gs1, finishing the analysis. \qed

\end{proof}

\begin{goodsit}\label{lem:gs5}
Given a partial packing such that an item $a$ with $s(a)>6$ is packed
into bin $A$, $s(B) \in [3,4)$, and $C$ is empty,
there exists an algorithm that packs all remaining items into
three bins of capacity $22$.
\end{goodsit}

\begin{proof}
%We have $s(B) \in [3,4)$, otherwise we reach \gs2.
Pack all incoming items into $A$ as long as it possible. %using \FF -- PV: IMO, FF is not needed to mention when packing into a single bin
If $s(A) \geq 12$, we have \gs4, and so we assume the contrary.
Therefore, $s(A) < 12$ and an item $i$ arrives which cannot be packed into $A$. 

Place $i$ into $B$. If $s(i) \geq 12$, we apply \gs3. We thus have
$s(i) \in (10,12)$ and $s(A) \geq 22 - s(i) > 10$. Continue with \FF
on bins $B$, $A$, and $C$ in this order.

We claim that \gs1 is reached at the latest after \FF has
packed two items, $x$ and $y$, on bins other than $B$. If one of them (say $x$) is
packed into bin $A$, this holds because $s(x) + s(B)>22$
and $s(A)>10$ already before this assignment -- enough for \gs1.
If both items do not fit in $A$, they are both larger than 10, since $s(A)<12$.
We will show by contradiction that this cannot happen.

As $s(A) < 12$ from our previous analysis, we note that $s(x), s(y) > 10$.
We therefore have three items $i,x,y$ with $s(i),s(x),s(y)>10$ and
an item $s(a) > 6$ from our initial conditions. These four items
cannot be packed together by any offline algorithm into three bins of
capacity 16, and so we have a contradiction with $s(x), s(y) > 10$.
\qed
\end{proof}

\begin{goodsit}\label{lem:gs6} 
If $s(C)<4$, $s(B)>6$ and $s(A)\ge s(B)+ 4-s(C)$, 
there exists an algorithm that packs all remaining items into
three bins of capacity $22$.
\end{goodsit}

\begin{proof}
Pack all items into $A$, until an item $x$ does not
fit.  At this point $s(A)+x>22$.  If $x$ fits on $B$, we put it there
and reach \gs1 because $s(B)>6$. Otherwise, $x$ definitely fits on $C$
because $s(C)<4$.  By the condition on $s(A)$, we have % $s(A)+s(C)-4+x \ge s(B)+x > 22$, so 
$x + s(A)+ s(C) \ge x + s(B)+4>26$, and we are in \gs1 again.
\qed
\end{proof}

%\begin{corollary}\label{cor:gs6}
%Suppose 
%$s(C)\le s(B)< 6< s(A)$. 
%If an item of size at least $s(A)-s(C)+4-s(B)$ arrives, we reach \gs6
%by placing it on $B$ or $C$.
%If an item of size at most $s(A)-s(C)-(4-s(B))$ arrives, it can be placed on $C$; we reach a good situation if $s(C)\ge4$ as a result.
%\end{corollary}
%
%\begin{proof}
%A large enough item can be assigned to $C$ (or $B$) and will bring us to the situation of the
%last lemma. A small enough item can be placed on $C$.
%\qed
%\end{proof}

\begin{goodsit}\label{lem:gs7}
Suppose 
$s(A)<4$, $s(C)<4$. % and $s(B)>6$. 
If $s(B)\le 9+\frac{1}{2}(s(A)+s(C))$ and for a new item $x$ we have
$s(B)+s(x)>22$,
then there exists an online algorithm that packs all remaining items into
three bins of capacity $22$.
\end{goodsit}

\begin{proof}
We have $s(B) > 22 - s(x) > 6$ and $s(x) > 22-s(B) \ge 13 - \frac{1}{2}(s(A)+s(C)) \ge s(B)+4-s(A)-s(C)$.
Placing $x$ on $A$ we increase the size of $A$ to at least $s(B)+4-s(C)$
and we reach \gs6. % by Corollary~\ref{cor:gs6}.
\qed
\end{proof}

\subsection{Good Situation First Fit}\label{sec:3alg}

Throughout our algorithm, we often use a special variant of \FF\/
which tries to reach good situations as early as possible. This
variant can be described as follows:

\begin{definition}\label{dfn:gsff}
Let $\calL=(X|_k, Y|_l, \ldots)$ be a list of bins $X,Y, \ldots $ where each
bin $X$ has an associated capacity $k$ satisfying $s(X)\leq k$.
% originally the limit $k$ was integral but it doesn't have to be
\GSFF{$\calL$} (Good Situation First Fit) is an online algorithm for bin stretching that works as follows:

\vspace{1.5ex}
\algobox{ %TODO PV this box might be written a nicer way
{\bf Subroutine \GSFF{$\calL$}:}
For each item $i$:\\
If it is possible to pack $i$ into any bin (including bins not in $\calL$, and
using capacities of 22 for all bins)
such that a good situation is reached,
do so and continue with the algorithm of the %\\
relevant good situation.\\
Otherwise, traverse the list $\calL$ in order and pack $i$ into the first bin $X$
such that $X|_k \in \calL$ and $s(X) + s(i) \leq k$.
If there is no such bin, stop.
}
\end{definition}

For example, \GSFF{$A|_4, B|_{22}$} checks whether
either $(A\cup \{j\},B,C)$, $(A,B\cup \{j\},C)$ or $(A,B,C\cup \{j\})$ is a partial packing of any good
situation. If this is not the case, the algorithm packs $j$ into bin
$A$ provided that $s(A)+s(j)\leq 4$. If $s(A) + s(j)>4$, the algorithm
packs $j$ into bin $B$ with capacity $22$. If $j$ cannot be placed
into $B$, \GSFF{$A|_4, B|_{22}$} halts and another
online algorithm must be applied to pack $j$ and subsequent items.

\subsection{The algorithm}\label{ssec:FP}

\vspace{1.5ex}

In a way, any algorithm for online bin stretching for three bins must
be designed so as to avoid several \textit{bad situations:} the two most
prominent ones being either two items of size $R/2$ or three items of
size $R/3$, where $R$ is the volume of the remaining items.

Our algorithm -- especially Steps \step{large} and \step{qstep} -- are
designed to primarily evade such bad situations, while making sure
that no good situation is missed. This evasive nature gives it its
name.

\def\itemskip{\hskip 3em}
\algobox{
\textbf{Algorithm} \tbalg:\\[-2.5ex]
\begin{compactenum}[(1)]
\item \label{fp:initial} Run \GSFF{$A|_4, B|_4$}.
\item \label{fp:rename} Rename the bins so that $s(A)\ge s(B)$.
\item \label{fp:jcheck} If the next item $j$ satisfies $s(j)>6$:
\item \label{fp:large}  \indentskip Set $p \defeq 6 + s(j)$; apply \GSFF{$A|_p, B|_{4}$}.
\item \label{fp:wcheck} \indentskip If the next item $w$ fits into $A|_{22}$:
\item \label{fp:wfits}  \indentskip \indentskip \GSFF{$A|_{22}, B|_{22}, C|_{22}$}.
\item \label{fp:otherwise} \indentskip Else: 
\item \label{fp:wwontfit}  \indentskip \indentskip \GSFF{$A|_p, B|_{22}, C|_{22}$}.
%Y(Recall that $p \defeq 6 + s(j)$.)

\item \label{fp:jelse} Else ($j$ satisfies $s(j) < 4$):
\item \label{fp:qstep} \indentskiptwodigit \GSFF{$A|_4,B|_q,C|_4$} where $q \defeq 9+\frac{1}{2}(s(A)+s(C))$.
\item[] \indentskip Update $q$ whenever $s(A)$ or $s(C)$ changes.
\item \label{fp:b1}\indentskiptwodigit\GSFF{$A|_4, B|_{22}, C|_{22}$}.
\item \label{fp:b2}\indentskiptwodigit\GSFF{$A|_{22}, B|_{22}, C|_{22}$}.
\end{compactenum}
}

\subsection{Analysis}

Let us start the analysis of the algorithm \tbalg in Step
\step{jcheck}, where the algorithm branches on the size of the item
$j$.

Throughout the proof, we will need to argue about loads of the 
bins $A$, $B$, $C$ before various items arrived. The following notation
will help us in this endeavour:

\medskip\noindent\textbf{Notation.} Suppose that $A$ is a bin and $x$ is an
item that gets packed at some point of the algorithm (not necessarily
into $A$).  Then $A_{\before x}$ will indicate the set of items that are
packed into $A$ just before $x$ arrived.

\smallskip
We first observe that our algorithm can be in two very different
states, based on whether $s(j) > 6$ or $s(j) < 4$.  Note that the case
$s(j) \in [4,6]$ is immediately settled using \gs2.

\begin{observation}
\label{obs:1a}
Assume that $s(j)<4$. We have $s(A_{\before j}) \in  (3,4)$
%, $s(B_{\before j}) + s(j) \le s(A_{\before j})+4$,
and $s(B_{\before j}) + s(j) \in(6,8)$
where $A$ and $B$ are bins after renaming in Step~{\rm \step{rename}}.
Thus both $A$ and $B$ received some items during Step~{\rm \step{initial}}.
Moreover, there is at most one item either in $A_{\before j}$, or in $B_{\before j}$.
\end{observation}

\begin{proof}
Since $s(j)<4$ and $s(B_{\before j})<4$, 
the item $j$ is assigned to $B$ in Step~\step{qstep}, which by Step~\step{rename} is
the least loaded bin among $A$ and $B$ after Step~\step{initial}.
For this bin, we have $s(B_{\before j}) > 2$, else 
we reach \gs2. % or \GSFF{$A|_4, B|_4$} does not fail on $j$.
This implies that both $A$ and $B$ received items in Step~\step{initial}, so $s(A_{\before j})+s(B_{\before j})>6$,
else a good situation would have been reached before $j$ arrived.
% We therefore have $s(B)<s(A)+4$ and $s(A)>3$ after renaming the bins.
It follows that $s(A_{\before j})\in(3,4)$ and $s(B_{\before j}) + s(j) \in (6,8)$.

Since any item that is put into $B$ during Step~\step{initial} must have size
of more than two (otherwise it fits into $A|_6$), only one such item can be packed into $B$
which proves the last statement.\qed
\end{proof}

Contrast the previous observation with the next one, which considers $s(j) > 6$:

\begin{observation}
\label{obs:1b}
Assume that $s(j) > 6$. Then, $s(A_{\before j}) < 3$, $s(B_{\before j}) = 0$.
\end{observation}

\begin{proof} If $s(A_{\before j}) \ge 3$ we reach \gs5 by packing
$j$ into $B$.  However, if $s(A_{\before j}) < 3$ then $s(A_{\before j}) + s(B_{\before j}) < 6$
which can be true (without reaching a good situation) only if $s(B_{\before j}) = 0$. \qed
\end{proof}

Both the analysis and the algorithm differ quite a lot based on the
size of $j$. If it holds that $s(j) > 6$, we enter the \emph{large
case} of the analysis, while $s(j) < 4$ will be analyzed as the
\emph{standard case}. Intuitively, if $s(j) > 6$, the offline optimum
is now constrained as well; for instance, no three items of size $10$
can arrive in the future. This makes the analysis of
the large case comparatively simpler.

\subsection{The large case}

We now assume that $s(j) > 6$. Our goal in both the large case and the
standard case will be to show that in the near future either a good
situation is reached or several large items arrive, but \tbalg is able
to pack them nonetheless.

Let us start by recalling the relevant steps of the algorithm:

\smallskip
\algobox{
\step{large} Set $p \defeq 6 + s(j)$; apply \GSFF{$A|_p, B|_{4}$}.\\
\step{wcheck} If the next item $w$ fits into $A|_{22}$:\\
\step{wfits}  \indentskip \GSFF{$A|_{22}, B|_{22}, C|_{22}$}.\\
\step{otherwise} Else: \\
\step{wwontfit}  \indentskip \GSFF{$A|_p, B|_{22}, C|_{22}$}
%, where $p \defeq 6 + s(j)$.
}

By choosing the limit $p$ to be $s(j)+6$ in Step \step{large}, we make
enough room for $j$ to be packed into $A$. We also ensure that
any item $i$ larger than $6$ that cannot be placed into $A$ with
capacity 22 must satisfy $s(i) + s(j) > 16$ and so $i$ cannot be with
$j$ in the same bin in the offline optimum packing.

Let us define $A_S$ as the set of items on $A$ of size less than $6$
(packed before or after $j$).  We note the following:

\begin{observation}\label{obs:nothingonc}

\begin{enumerate}
\item During Step {\rm \step{large}}, if $B$ contains any item, it is true that
$s(A_S) + s(B) > 6$.
\item If no good situation is reached, the item $w$ ending Step {\rm \step{large}} satisfies $s(w) > 6$.
\end{enumerate}

\end{observation}

\begin{proof}
The first point follows immediately from our choice of $p$ and \GSFF{$A|_p, B|_4$}.

For the second part of the observation, consider the item $w$ that
ends Step~\step{large} and assume $s(w) \le 6$. The possibility that
$s(w) \in  [4,6]$ is excluded due to \gs2. The case $s(B_{\before w}) \ge  3$ is
also excluded, as this would imply \gs5 with $j$ in $A$.

Since $s(B_{\before w}) + s(w) > 6$, the only remaining possibility is $s(B_{\before w}) \in  [2,3), s(w) \in 
(3,4)$. Even though $w$ does not fit into $A|_p$, if we were to pack
$w$ into $A|_{22}$, we can use the first point of this observation and get
$s(A) + s(B) \ge  s(j) + \big(s(A_S) + s(B_{\before w})\big) + s(w) > 6 + 6 + 3 = 18$, enough for \gs4. The algorithm
\GSFF{$A|_{p}, B|_{4}$} in Step \step{large} will notice this
possibility and will pack $w$ into $A$, where it will always fit, as
$s(A_{\before w}) < 15$ by \gs4. \qed

\end{proof}

We now split the analysis based on which branch is entered in Step \step{wcheck}:

\smallskip
\noindent
{\bf Case 1:} Item $w$ fits into bin $A$; we enter Step \step{wfits}.

We first note that $s(A)+s(B)<15$, else we are in \gs4 since $C$ is still empty. This inequality also
implies that $s(B) = 0$, otherwise we have $s(A) + s(B) = s(w) + s(j) + (s(A_S) +s(B)) > 18$
via Observation~\ref{obs:nothingonc}.

We continue with Step~\step{wfits} until we reach a good situation or
the end of input.  Suppose three items $x,y,z$ arrive such that none of them 
can be packed into $A$ and we do not reach a good situation. We will prove that
this cannot happen. We make several quick observations about those items:

\begin{enumerate}
\item We have $s(x) > 7$ because $s(A_{\before x}) < 15$ or we reach \gs4. The item $x$ is packed into $B$.
\item At any point, $B$ contains at most one item, otherwise $s(A) + s(B) > 22 + 7 > 26$, reaching \gs1.
\item We have $s(y) > 9$ because $\min(s(A_{\before y}), s(B_{\before y})) < 13$ by \gs1. The item $y$ is packed into $C$.
\item The bin $C$ contains also at most one item, similarly to $B$.
\item Again, we have $s(z) > 9$ similarly to $y$. The item $z$ does not fit into any bin.
\end{enumerate}

From our observations above, we get $s(x) + s(y) > 22$, $s(x) + s(z) >
22$, $s(y) + s(z) > 22$. Therefore, at least two of the items
$\{x,y,z\}$ are of size at least $11$. However, both items $j$ and $w$ have size at least $6$,
and there is no way to pack $j,w$ and the two items larger than $11$ into
three bins of capacity $16$, a contradiction.

\smallskip
\noindent \textbf{Case 2:} Item $w$ does not fit into bin $A|_{22}$. The choice
of $p$ gives us $s(j)+s(w)>16$. Item $w$ is placed on $B$.

The limit $p$ gives us an upper bound on the volume of small items
$A_S$ in $A$, namely $s(A_S) \le 6$.  An easy argument gives us a similar bound on $B$,
namely if $B_S \defeq B \setminus \{w\}$, then $s(B_S) < 4$. Indeed, we have
$26 > s(A) + s(B) > 22 + s(B_S)$, the first inequality implied by not
reaching \gs1.

In Case 2, it is sufficient to consider two items $x,y$ that do not
fit into $A|_{p}$ or $B|_{22}$. We have:

\begin{enumerate}
\item Using $s(B_S)<4$, we have $s(x) + s(w) > 18$ and $s(y) + s(w) > 18$.
\item None of the items $x,y$ fits into $A|_{22}$. If say $x$ did fit, then we
use the fact that $x$ does not fit into $B|_{22}$ and get
$s(B) + s(A) = \big(s(B_{\before x}) + s(x)\big) + s(A_{\before x}) > 22 + s(j) > 26$ and we reach \gs1.
\item The choice of the limit $p$ on $s(A)$ implies $s(x) + s(j) > 16$ and $s(y) + s(j) > 16$.
\item Since $\min(s(A),s(B)) < 13$ at all times by \gs1, we have $s(x)>9$ and $s(y) > 9$.
\item The items $x$ and $y$ do not fit together into $C$, or we would have
$s(C) + s(A) > 22 + s(y) > 26$.  This implies $s(x) + s(y) > 22$.

\end{enumerate}

From the previous list of inequalities and using $s(j) + s(w) > 16$,
we learn that no two items from the set $\{j,w,x,y\}$ can be together
in a bin of size $16$. Again, this is a contradiction with the
assumptions of \binstretch.

\subsection{The standard case}

From now on, we can assume that $s(j) < 4$, $j$ is packed into $B$ and
Step \step{qstep} of \tbalg is reached. Recall that by Observation~\ref{obs:1a}
$s(A_{\before j}) \in  (3,4)$, $s(B_{\before j}) + s(j) \in(6,8)$, and
there is exactly one item either in $A$, or in $B$; we denote this item by $e$.
We repeat the steps done by \tbalg in the standard case:

\smallskip
\algobox{
\step{qstep} \GSFF{$A|_4,B|_q,C|_4$} where $q \defeq 9+\frac{1}{2}(s(A)+s(C))$.\\
\-\hspace{2em} Update $q$ whenever $s(A)$ or $s(C)$ changes.\\
\step{b1} \GSFF{$A|_4, B|_{22}, C|_{22}$}.\\
\step{b2} \GSFF{$A|_{22}, B|_{22}, C|_{22}$}.
}

Recall that $s(A) > 3$ by Observation~\ref{obs:1a}. Assuming
that no good situation is reached before Step \step{qstep}, we
observe the following:

% \smallskip

\begin{observation}
\label{obs:6-b}
In Step {\rm \step{qstep}}, as long as $C$ is empty,
packing any item of size at least $4$ leads to a good
situation. Thus while $C$ is empty, all items that 
arrive in Step {\rm \step{qstep}} and are not put on $A$ 
have size in $(6-s(A),4)$.
\end{observation}
\begin{proof}
Any item with size in $[4-s(A),6-s(A)]\cup[4,6]$ leads to \gs2. 
Any item with size more than 6 is assigned to $B$ if it fits there,
reaching \gs5, and else to $A$ or $C$, reaching \gs7 since $s(B)\le9+\frac{1}{2}(s(A)+s(C))$. The only remaining
possible sizes of items that are not packed into $A$ are $(6-s(A),4)$.\qed
\end{proof}

\begin{corollary}
\label{obs:c}
After Step {\rm \step{qstep}}, $C$ contains exactly one item $r$ and $s(A)+s(C)>6$.
\end{corollary}

\begin{proof}
From the previous observation it is clear that $C$ receives at least one item $r$ in Step \step{qstep}. No second
item $r_2$ can be packed into $C|_4$ in Step~\step{qstep} as $s(r) + s(r_2)$ would be at least $2(6 - s(A)) > 4$. \qed
\end{proof}

Step \step{qstep} terminates with a new item $x$ which fits into %PV: observation for further reference?
$B|_{22}$ (otherwise we would reach \gs7),
but not below the limit $q = 9 + \frac{1}{2}(s(A)+s(C))$. We pack $x$
into $B$ in Step~\step{b1}, getting $s(B) > 9 + \frac{1}{2}(s(A)+s(C)) > 12$.

A possible bad situation for our current packing is when three items $b_1,
b_2, b_3$ arrive, where the items are such that no two items
of this type fit together into any bin, and no single item of this
type fits on the largest bin, which is $B$ in our case. In fact, we
will prove later that this is the only possible bad situation.

We claim that this potential bad situation cannot
occur:

\begin{clm}\label{clm:nobigitems}
Suppose that algorithm \tbalg reaches no good situation in the standard case. Then,
$s(C)\ge s(r)>2.8$ and after placing $x$ into B in Step~{\rm \step{b1}} it holds that
$s(B)<12.8$.

Furthermore, suppose that among items that arrive after $x$, there are three items $b_1, b_2, b_3$ such that $\min( s(b_1), s(b_2), s(b_3)) > 8$. Then, it holds that 
\[\min(s(b_1),s(b_2),s(b_3)) < 9.2. \]

% Martin: s(r) > 2.8 is not actually the best bound we can derive; later we can show that s(r)>3.
% However, s(r) > 2.8 is trivially true and it is enough for our calculations.
\end{clm}

We now show how Claim \ref{clm:nobigitems} finishes the analysis of \tbalg.

After Step \step{qstep}, assuming no good situation was reached, the
algorithm places $x$ into $B|_{22}$ and continues with Step~\step{b1},
which is \GSFF{$A|_4, B|_{22}, C|_{22}$}. Claim \ref{clm:nobigitems}
gives us that $s(B) < 12.8$ after placing $x$, while the fact that we exited
Step \step{qstep} means that $s(B) > q = 9 + \left( s(A)+s(C) \right)/2 > 12$.

Consider the first item $b_1$ that does not fit into $A|_4$. We have
that $s(b_1)>2$, otherwise \gs2 is
reached. However, any item that fits into $B$ (as long as $s(C) \le 4$) triggers
\gs4, because $s(A) + s(B) + s(b_1) \ge  6 + 12 > 15 + s(C)/2$.

We now know that the first item $b_1$ does not fit into both $A|_{4}$ and $B|_{22}$.
We place it into $C$, noting that $s(b_1) > 22 - s(B) \ge  22 - 12.8 = 9.2$.

We keep packing items into $A|_{4}$, waiting for the second item $b_2$
that does not fit into $A|_4$ in Step~\step{b1}. Again, $s(b_2) >
2$. Suppose that $b_2$ fits into $B|_{22}$ or $C|_{22}$. Claim \ref{clm:nobigitems}
gives us $s(r)>2.8$; we thus sum up bins $B$ and $C$ and get $s(B_{\before b_2}) +
s(b_2) + s(r) + s(b_1) > 12 + 2 + 2.8 + 9.2 = 26$, which is enough for
\gs1.  Our assumption was false, the item $b_2$ does 
fit into neither $B|_{22}$ nor $C|_{22}$, in particular $s(b_2)>9.2$.

We move to Step~\step{b2}, pack $b_2$ into $A|_{22}$ and continue
with \GSFF{$A|_{22}, B|_{22}, C|_{22}$}. If at any time $s(A) \ge  14$, we
enter \gs1 on $A$ and $B$. Otherwise, if an item $b_3$ does not fit into $A|_{22}$,
it must satisfy $s(b_3) > 8$.

We now apply the full strength of Claim \ref{clm:nobigitems}. The
smallest item of $b_1$, $b_2$, $b_3$ must have size less than $9.2$,
and because of our argument, it must be $b_3$ -- but this means it
fits into $B$, as $s(B)<12.8$. \gs1 on bins $A$ and $B$ finishes the
packing, since $s(A) + s(b_3) + s(B_{\before b_3}) > 22 + 12 > 26$.

\subsection{Proof of Claim \ref{clm:nobigitems}}

Our current goal is to prove Claim \ref{clm:nobigitems}.
As in the large case, we would now like to appeal to the offline
layout of the larger items currently packed. Unlike the large case,
none of the items we have packed before Step~\step{b1} is guaranteed to be over $6$.

Sidestepping this obstacle, we will argue about the offline layout of
the smaller items. We now list several items that are packed before Step~\step{b2} and will be important in our analysis:

\begin{definition}\label{dfn:fouritems} The four items $e,j,r,x$ are defined as follows:
\begin{enumerate}
\item The item $e, 2 < s(e) < 4$: the only item packed into $B$ in Step~{\rm \step{initial}}
by Observation~\ref{obs:1a}.
(Note that $e$ might end up on $A$ after renaming the bins.)
\item The item $j, 2 < s(j) < 4$, defined in Step~{\rm \step{jcheck}}.
\item The item $r, 2 < s(r) < 4$, which is placed into $C$ in Step~{\rm \step{qstep}} by Observation~\ref{obs:c}; $r$ is the only item in $C$ until Step~{\rm \step{b1}}.
\item The item $x$ which terminated Step~{\rm \step{qstep}}.
\end{enumerate}

\end{definition}

There are four such items and only three bins, meaning that in the
offline optimum layout with capacity $16$, two of them are packed in
the same bin.  We will therefore argue about every possible pair,
proving that each pair is of size more than $6.8$.

Our main tool in proving the mentioned lower bounds are the
inequalities that must be true during various stages of algorithm
\tbalg, since a good situation was not reached. We now list
all the major inequalities that we will use:

\begin{itemize}

\item At the beginning of Step~\step{qstep}, packing $j$ makes the bin $B$ go over $6$:
\begin{equation}\label{eq:1}
s(B_{\before j}) + s(j) > 6.
\end{equation}

\item In Step~\step{qstep}, we know that any item $i$ that gets packed into $B$ (and does not fit into $A|_4$) also does not fit into $A|_{6}$ or else we reach \gs2:
\begin{equation}\label{eq:2}
s(A_{\before i}) + s(i) > 6.
\end{equation}

\item The item $r$ does not fit into $B|_q$:
\begin{equation}\label{eq:3}
s(B_{\before r})+ s(r) > 9 + \frac{s(A_{\before r})}{2}.
\end{equation}

\item The item $r$ cannot cause \gs4 if it is packed into $B$ (with $C$ still empty):

\begin{equation}\label{eq:4}
\left(s(B_{\before r}) + s(r) \right) + s(A_{\before r}) < 15.
\end{equation}

\item When $x$ got packed into $B|_{22}$ in Step~\step{b1}, it did not
cause \gs4 when summing $B$ with $C$:

\begin{equation}\label{eq:5}
\left(s(B_{\before x}) + s(x) \right) + s(r) < 15 + \frac{s(A_{\before x})}{2}. 
\end{equation}

\item The item $x$ also did not cause \gs4 when summing $B$ with $A$:

\begin{equation}\label{eq:6}
\left(s(B_{\before x})+ s(x) \right) + s(A_{\before x}) < 15 + \frac{s(r)}{2}. 
\end{equation}

\item Packing $x$ into $B|_{22}$ in Step~\step{b1} causes $B$ to go over the limit $q$:
\begin{equation}\label{eq:7}
s(B_{\before x}) + s(x) > 9 + \frac{s(A_{\before x}) + s(r)}{2}.
\end{equation}

\item The item $x$ cannot start \gs6 when packed into $A$. Comparing bin $A$ to $B$, we get:
\begin{equation}\label{eq:8}
s(A_{\before x}) + s(x) < s(B_{\before x}) + (4 - s(r)).
\end{equation}

We get a similar but slightly different inequality when comparing $B$ to $A$ instead:
\begin{equation}\label{eq:8b}
s(B_{\before x}) < s(A_{\before x}) + s(x) + \big(4 - s(r)\big).
\end{equation}

\item \gs6 could not be reached when the algorithm
considered packing $x$ into bin $C$, comparing the bin $C$
to $B$:
\begin{equation}\label{eq:9b}
s(r) + s(x) < s(B_{\before x}) + \big(4 - s(A_{\before x})\big).
\end{equation}

Again as in inequality \eq{8b} we can compare bin $B$ to $C$ and get:
\begin{equation}\label{eq:9}
s(B_{\before x}) < s(r) + s(x) + (4 - s(A_{\before x})). 
\end{equation}

% \item Placing $x$ on $A_{\before x}$ caused $A$ to go over $q$:
% \[s(A_{\before j}) + s(j) + s(x) > 9 + \frac{s(B_{\before x})  + s(r)}{2}, \]
\end{itemize}

% \[ s(A_{\before j}) + s(j) + s(x) + s(B_{\before x}) < 15 + \frac{s(r)}{2}. \]

% On the other hand, placing $x$ into $C$ did not cause \gs6, giving us

% \[ s(r) + s(x) < s(A_{\before j}) + s(j) + 4 - s(B_{\before x}). \]

% Already in the list.
% \item Placing $r$ into $B_{\before r}$ would also cause $A$ to go over $q$:
% \begin{equation}\label{eq:eplusr2}
% s(A_{\before j}) + s(j) + s(r) > 9 + \frac{s(B_{\before r})}{2},
% \end{equation}

Note that one can prove Claim~\ref{clm:nobigitems} by showing that the
converse of the claim and the above inequalities form an infeasible
linear programming instance. This was our approach as
well. Nonetheless, we provide an explicit proof for completeness.

Our first lemma establishes that $j$ is actually the only item that
is packed into $B$ during Step~\step{qstep}, which intuitively means that
$j$ is not too small:

\begin{lemma}\label{lem:onej}
Assume that no good situation is reached until Step~{\rm \step{b1}}. Then
it holds that during Step~{\rm \step{qstep}}, only $j$ is packed into $B$.
\end{lemma}

\begin{proof}

We first prove that no two additional items $j_2, j_3$ can be packed into $B$
during Step~\step{qstep}. Assuming the contrary, we get $s(B_{\before j_2}) +
s(j_2) + s(j_3) > 6 + 2 + 2 = 10$. With that load on $B$, we consider
the packing at the end of Step~\step{qstep}, when the item $x$ arrived. If
$s(x) + s(C_{\before x}) < 9$, we get \gs6 by placing $x$ into $C$ since $s(A)>3$, so it must be true
that $s(x) + s(C_{\before x}) > 9$, which means $s(x) >
5$. This is enough for us to place $x$ into $B|_{22}$
(where it fits, otherwise we are in \gs7) and reach \gs3.

This contradiction gives us that at most one additional item $j_2$ can
be packed into $B$ during Step \step{qstep}. We will now prove that
even $j_2$ does not exist.

% \begin{itemize}

% \item We know that packing $j$ makes the bin $A$ go over $6$:
% \begin{equation}\label{eq:1}
% s(A_{\before j}) + s(j) > 6.
% \end{equation}

% \item We know that $j_2$ could not fit into $B|_{6}$:
% \begin{equation}\label{eq:2}
% s(B_{\before j_2}) + s(j_2) > 6.
% \end{equation}

% \item \gs6 could not be reached when the algorithm
% considered packing $x$ into bin $C$ which contains only $r$:
% \begin{equation}\label{eq:9}
% s(A_{\before j}) + s(j) + s(j_2) < s(r) + s(x) + (4 - s(B_{\before x})). 
% \end{equation}

% \end{itemize}

We split the analysis into two cases depending on which of $j_2$ and $r$ arrives first.

\smallskip
\noindent \textbf{Case 1.} The item $r$ is packed
before $j_2$, meaning $s(B_{\before x}) = s(B_{\before r}) + s(j_2)$.

% \begin{itemize}
% \item The item $r$ does not fit into $A|_q$:

% \begin{equation}\label{eq:3}
% s(A_{\before j}) + s(j) + s(r) > 9 + \frac{s(B_{\before r})}{2}.
% \end{equation}

% \item When $x$ got packed into $A$, it did not cause \gs4 (with bins $B$ and $C$ swapped):

% \begin{equation}\label{eq:5}
% \left(s(A_{\before j}) + s(j) + s(j_2) + s(x) \right) + s(r) < 15 + \frac{s(B_{\before x})}{2}. 
% \end{equation}

% \end{itemize}

We start with inequalities \eq{3}, \eq{6} and \eq{8b} in the following form:

\begin{align*}
9 + \frac{s(A_{\before r})}{2} &< s(B_{\before r}) + s(r) \\
s(B_{\before r}) + s(j_2) + s(x) + s(A_{\before x}) &< 15 + \frac{s(r)}{2} \\
s(B_{\before r}) + s(j_2) &< s(A_{\before x}) + s(x) + \big(4 - s(r)\big)
\end{align*}

We sum twice \eq{3} with \eq{6} and \eq{8b}:

\begin{align*}
&18 + s(A_{\before r}) + s(B_{\before r}) + s(j_2) + s(x) + s(A_{\before x}) + s(B_{\before r}) + s(j_2) \\
&< 2s(B_{\before r}) + 2s(r) + 15 + \frac{s(r)}{2} + s(A_{\before x}) + s(x) + 4 - s(r)
\end{align*}
\[ s(A_{\before r}) + 2s(j_2) < \frac{3s(r)}{2} + 1\]

Using $s(A_{\before r}) \ge s(A_{\before j})$ (since $r$ arrives after $j$)
along with $s(A_{\before j}) > 3$ from Observation \ref{obs:1a} and $s(r) < 4$ gives us:

\begin{align*}
3 + 2s(j_2) &< 7 \\
s(j_2) &< 2
\end{align*}

which is a contradiction, since $s(A_{\before j_2}) < 4$ and $j_2$ did not fit into $A|_{6}$.

\smallskip

\noindent \textbf{Case 2.} In the remaining case, $j_2$ arrives before $r$, which means $s(B_{\before x}) = s(B_{\before r}) = s(B_{\before j}) + s(j) + s(j_2)$.

% \begin{itemize}
% \item The item $r$ cannot cause a \gs4 if it
% is placed on $B$ (note that we need to know that $j_2$ arrived before $r$):

% \begin{equation}
% \left(s(B_{\before j}) + s(j) + s(j_2) + s(r) \right) + s(A_{\before r}) < 15.
% \end{equation}

% \item When $x$ got packed into $B$, it also did not cause \gs4:

% \begin{equation}
% \left(s(B_{\before j}) + s(j) + s(j_2) + s(x) \right) + s(A_{\before x}) < 15 + \frac{s(r)}{2}. 
% \end{equation}

% \end{itemize}

We start by summing \eq{4} and \eq{6}. We get:
\begin{align}
s(B_{\before r}) + s(B_{\before x}) + s(x) + s(r) + s(A_{\before x}) + s(A_{\before r}) &< 30 + s(r)/2 \nonumber \\ 
2s(B_{\before j}) + 2s(j) + 2s(j_2) + s(x) + s(r) + s(A_{\before r}) + s(A_{\before x}) &< 30 + s(r)/2. \label{eq:result1}
\end{align}

Keeping \eq{result1} in mind for later use, we continue
by considering \eq{1}, \eq{2} and \eq{9} in the
following form:

\begin{align}
s(B_{\before j}) + s(j) &> 6 \label{eq:1applied}\\
s(A_{\before j_2}) + s(j_2) &> 6 \label{eq:2applied}\\
s(r) + s(x) + (4 - s(A_{\before x})) &> s(B_{\before x}) = s(B_{\before j}) + s(j) + s(j_2) \nonumber
\end{align}

Summing the three inequalities gives us:
\[ s(B_{\before j}) + s(j) + s(A_{\before j_2}) + s(j_2) + s(r) + s(x) + (4 - s(A_{\before x})) > 12 + s(B_{\before j}) + s(j) + s(j_2) ,\]
\[s(r) + s(x) + \left( s(A_{\before j_2}) - s(A_{\before x}) \right) > 8, \]
\begin{equation}\label{eq:result3}
s(r) + s(x) > 8.
\end{equation}

Summing two times \eq{1applied}, two times
\eq{2applied} and once \eq{result3} gives us:

\begin{equation}\label{eq:result2}
2s(B_{\before j}) + 2s(j) + 2s(j_2) + 2s(A_{\before j_2}) + s(r) + s(x) > 32.
\end{equation}

Using $s(A_{\before j_2}) \le s(A_{\before r}) \le s(A_{\before x})$ (which is only true here
in Case 2, where $r$ arrived later) and recalling
\eq{result1} along with \eq{result2}, we get $30 +
s(r)/2 > 32$ and $s(r) > 4$, which is a contradiction with $r$ fitting into
$C|_4$. \qed

\end{proof}

Having established that only one item $j$ is packed into $B$ during
Step \step{qstep}, we can start deriving lower bounds on pairs of
items from the set $\{e,j,r,x\}$. We will prove these bounds similarly
to Lemma \ref{lem:onej}, mostly by summing bounds that arise from
evading various good situations.

\begin{lemma}\label{lem:eplusr}
Suppose that $e$ and $r$ are items as described in Definition \ref{dfn:fouritems} and suppose
also that no good situation was reached during Step~{\rm \step{qstep}} of the algorithm \tbalg.
Then, $s(e) + s(r) \ge  s(B_{\before j}) + s(r) > 6.8$.
\end{lemma}

\begin{proof}
First of all, it is important to note that the item $e$ may be
packed on $A$ or on $B$.
% Regardless, if we prove $s(B_{\before j}) + s(r) >
% 6.8$, we will be done: $j$ is placed on top of $B_{\before j}$
% and we have $s(B_{\before j}) \le s(A_{\before j})$ by Step~\step{rename}.
Since either $B_{\before j}$, or $A_{\before j}$ contains solely $e$
by Observation~\ref{obs:1a}, we get that either $s(B_{\before j}) = s(e)$,
or $s(B_{\before j}) \le s(A_{\before j}) = s(e)$. Thus it is sufficient to prove
$s(B_{\before j}) + s(r) > 6.8$.

We start the proof of $s(B_{\before j}) + s(r) > 6.8$ by restating \eq{3}, \eq{7},
and \eq{8} in the following form:

\begin{align*}
s(B_{\before j}) + s(j) + s(r) &> 9 + \frac{s(A_{\before r})}{2} \\
s(B_{\before j}) + s(j) + s(x) &> 9 + \frac{s(A_{\before x}) + s(r)}{2} \\
s(B_{\before j}) + s(j) + (4 - s(r)) &> s(A_{\before x}) + s(x). 
\end{align*}

Before summing up the inequalities, we multiply the first one by 8, the
second by 2 and the third by 2. % The multiplied equations are:
%
% \begin{align*}
% 2s(B_{\before j}) + 2s(j) + 2s(x) &> 18 + s(A_{\before x}) + s(r) \\
% 8s(B_{\before j}) + 8s(j) + 8s(r) &> 72 + 4s(A_{\before r}) \\
% 2s(B_{\before j}) + 2s(j) + 8 - 2s(r)) &> 2s(A_{\before x}) + 2s(x). 
% \end{align*}
%
In total, we have:
\[ 12s(B_{\before j}) + 12s(j) + 8 + 6s(r) + 2s(x) > 90 + 3s(A_{\before x}) + 4s(A_{\before r}) + s(r) + 2s(x). \]

We know that $s(B_{\before j}) \leq s(A_{\before x})$ and $s(B_{\before j}) \leq s(A_{\before r})$, %giving us $-s(B_{\before j}) \geq -s(A_{\before x})$,
% $-s(B_{\before j}) \geq -s(A_{\before x})$,
allowing us to cancel out the terms:

\[ 5s(B_{\before j}) + 5s(r) + 12s(j) > 82. \]

Finally, using the bound $s(j) < 4$ and noting that $(82 - 48)/5 = 6.8$, we get

\[ s(B_{\before j}) + s(r) > 6.8. \qquad \qed\]
\end{proof}

\begin{lemma}\label{lem:eplusj}
Suppose that $e$ and $j$ are items as described in Definition \ref{dfn:fouritems} and suppose
also that no good situation was reached by the algorithm \tbalg.
Then, $s(e) + s(j) \ge  s(B_{\before j}) + s(j) > 7.6$.
\end{lemma}

\begin{proof}

The same argument as in Lemma \ref{lem:eplusr} gives us $s(e) + s(j) \ge  s(B_{\before j}) + s(j)$.
We therefore aim to prove $s(B_{\before j}) + s(j) > 7.6$. Summing up \eq{7}
and \eq{9b} and using $s(B_{\before x}) = s(B_{\before j}) + s(j)$, we get

\[ 2s(B_{\before j}) + 2s(j) + s(x) + 4 - s(A_{\before x}) > 9 + \frac{s(A_{\before x}) + s(r)}{2} + s(r) + s(x)\]
\[ 2s(B_{\before j}) + 2s(j) > 5 + \frac32 \big(s(A_{\before x})  + s(r)\big).\]

We now apply the bound $s(A_{\before x}) + s(r) \ge  s(B_{\before j}) + s(r) > 6.8$, the second inequality being Lemma \ref{lem:eplusr}. We get:

\[ 2s(B_{\before j})+ 2s(j) > 5+ 10.2, \]
and finally $s(B_{\before j}) + s(j) > 7.6$, completing the proof.\qed

\end{proof}

\begin{lemma}\label{lem:rplusj}
Suppose that $j$ and $r$ are items as described in Definition \ref{dfn:fouritems} and suppose
also that no good situation was reached by the algorithm \tbalg. Then, $s(r) + s(j) > 7$.
\end{lemma}

\begin{proof}
Starting with \eq{3}:

\[ s(B_{\before j}) + s(j) + s(r) > 9 + \frac{s(A_{\before r})}{2}  \]
and using $s(B_{\before j})\le s(A_{\before j})\le s(A_{\before r})$ together with $s(B_{\before j}) < 4$, we have:
\[ s(j) + s(r) > 9 + \left( \frac{s(A_{\before r})}{2} - s(B_{\before j}) \right) \ge  9 - \frac{s(B_{\before j})}{2} > 7. \qquad \qed \]

\end{proof}

\begin{lemma}\label{lem:xplus}
Suppose that $x,e,j,r$ are items as described in Definition \ref{dfn:fouritems}.
Suppose also that no good situation was reached by the algorithm \tbalg.
 Then, $s(x) > 4$ and $s(x) + \min(s(j),s(e),s(r)) > 6.8$.
\end{lemma}

\begin{proof}
With all the previous lemmas in place, the proof is simple enough. We first
observe that $s(x) > 4$; this is true because $s(B_{\before j}) + s(j) < 4+4 = 8$ and
$s(B_{\before x}) + s(x) > q ≥ 12$.

Since the remaining three items $\{e,r,j\}$ are bounded from above
by $4$ but their pairwise sums are always at least $6.8$, we have
that $\min\{e,r,j\} > 2.8$, which along with $s(x) > 4$ gives us the
required bound.\qed
\end{proof}

From Lemmata \ref{lem:eplusr}, \ref{lem:eplusj}, \ref{lem:rplusj} and
\ref{lem:xplus} we get a portion of Claim \ref{clm:nobigitems}: if
three big items $b_1, b_2, b_3$ exist in the offline layout, then one of these
items needs to be packed together with at least two items from the set
$\{e,j,r,x\}$, and therefore $\min(s(b_1),s(b_2),s(b_3)) < 9.2$. The
second bound $s(r) > 2.8$ follows from Lemma \ref{lem:eplusr} and the
fact that $s(e) < 4$.

All that remains is to prove the bound on $s(B)$, which we do in the following lemma:

\begin{lemma}\label{lem:boundonB}
Suppose that no good situation was reached in the algorithm \tbalg during Step~{\rm \step{qstep}}.
Then, after placing $x$ into $B$ in Step~{\rm \step{b1}}, it holds that $s(B) < 12.8$.
\end{lemma}

\begin{proof}

As before, we will use our inequalities to derive the desired
bound. As we have argued above, Lemma \ref{lem:eplusr} gives us that
$s(r) > 2.8$.

We sum up inequalities \eq{6} and \eq{9b}, getting:

\begin{align*}
 s(B_{\before j}) + s(j) + 2s(x) + s(A_{\before x}) + s(r) &< 15 + \frac{s(r)}{2} + s(B_{\before j}) + s(j) + 4 - s(A_{\before x})\\ 
2s(x) + 2s(A_{\before x}) &< 19 - \frac{s(r)}{2}. \\
s(x) + s(A_{\before x}) &< 9.5 - \frac{s(r)}{4}. 
\end{align*}
To finish the bound we need $s(B_{\before j}) \le s(A_{\before j}) \le s(A_{\before x})$ (this is true because we reorder the bins
$B$, $A$ in Step~\step{rename}), $s(r) > 2.8$ and $s(j) < 4$. Plugging them in, we get:

\[ s(B) = s(B_{\before j}) + s(j) + s(x) \le s(A_{\before x}) + s(j) + s(x) \]
\[ < 9.5 - \frac{s(r)}{4} + s(j) < 9.5 - 0.7 + 4 < 12.8. \qquad \qed\]
\end{proof}

With Lemma \ref{lem:boundonB} proven, we have finished the proof of Claim \ref{clm:nobigitems}
and completed the analysis of the algorithm \tbalg.

% --- SNIP ---
\section{Lower bound}\label{sec:lowerbound}

In this section, we describe our lower bound technique for a small
number of bins. We build on the paper of Gabay, Brauner and Kotov
\cite{gabay2013lbv2} but significantly change the algorithm, both
conceptually and technically.

On the conceptual side, we propose a different algorithm for computing
the offline optimum packing, suggest new ways of pruning the game tree
and show how the alpha-beta pruning of \cite{gabay2013lbv2} can be
skipped entirely.

On the technical side, we reimplement the algorithm of
\cite{gabay2013lbv2}, gaining significant speedup from the
reimplementation alone. While the lower bound search program of
\cite{gabay2013lbv2} was written in Python, employed CSP solvers and had
unrestricted caching, our program is written in C, is purely
combinatorial and it sets limits on the cache size, making time the
only exponentially-increasing factor.

%TODO: check after recomputing results

With these improvements, we were able to find an improved lower bound
for \binstretch for three bins, namely $45/33 = 1.\overline{36}$.

We also present the lower bound of $19/14 \approx 1.357$ for $m = 4$ and $m = 5$.
Note that this is the first non-trivial lower bound for $m = 5$ and that
our result is independent from the lower bound of $19/14$ for $m = 4$
by Gabay et al.~\cite{gabay2013lbv2}.
\smallskip

To see the strength of our improvements, consider the scaling factor $K$
and items of integer size. It is easy to see that a general game
tree search requires exponential running time with respect to
$K$. The algorithm of \cite{gabay2013lbv2} is able to check all $K \le 20$
(for $m=3$) before claiming that ``even with many efficient cuts, we cannot tackle
much larger problems.''

In contrast, our proposed algorithm is able to check all $K \le 41$ and
is fast enough to produce results for $m = 4$ and $m = 5$.

\subsection{Lower bound technique}\label{sec:technique}

We now describe our lower bound technique. To simplify our arguments,
we describe the technique only for $m = 3$. We discuss the
pecularities of the generalization to any fixed $m$ in
Section~\ref{sec:generalization}.

As with many other online algorithms, we can think of \binstretch as a
two player game. The first player (\algo) is presented
with an item $i$. \algo's goal is to pack it into $m$ bins of capacity
$S$. This mimics the task of any algorithm for \binstretch. The other
player (\adversary) decides which item to present to the \algo in the
next step. The goal of the \adversary is to force \algo to overpack at
least one bin.

It is clear that knowing the game tree for a parameter $S$ of the
aforementioned game is equal to knowing whether there is an algorithm
for \binstretch with stretching factor $S$.

We are interested primarily in the lower bound. Therefore, it makes
sense to slightly reformulate the previous game:

\begin{itemize}
\item The player \algo wins if it can pack all items into bins with capacity strictly less than $S$.
\item The player \adversary wins if it can force \algo to pack a bin with load $\ge S$ while making
  sure that the \binstretch guarantee is satisfied.
\end{itemize}

This way, a winning strategy for the player \adversary immediately implies that
no online algorithm for \binstretch with stretching factor less than
$S$ exists.

The two main obstacles to implementing a search of the described two player
game are the following:

\begin{enumerate}
\item \adversary can send an item of arbitrarily small size;
\item \adversary needs to make sure that at any time of the game, an offline
  optimum can pack the items arrived so far into three bins of size $T$.
\end{enumerate}

To overcome the first problem, it makes sense to create a sequence of
games based on the granularity of the items that can be packed. A
natural granularity for the scaled game are integral items, which
correspond to multiples of $1/T$ in the non-scaled problem.

The second problem increases the complexity of every game turn of the
\adversary, as it needs to run a subroutine to verify the guarantee
for the next item it wishes to place.

Note that the ideas described above have been described previously in \cite{gabay2013lbv2}.
\medskip

To precisely formulate our setting, we first define one state of a game:

\begin{definition}\label{dfn:binconf}
  For given parameters $S \in \Nat, T \in \Nat$, a \textbf{bin configuration} is
  a tuple $(a,b,c,\It)$, where
  \begin{compactitem}
    \item $a,b,c \in \{0,1,\ldots,S\}$ denote the current sorted loads of the bins, i.e., $a\ge b\ge c$,
    \item $\It$ is a multiset with ground set $\{1,2,\ldots,T\}$ which lists the items used in the bins.
  \end{compactitem}

  Additionally, in a bin configuration, it must hold:
  \begin{compactitem}
    \item that there exists a packing of items from $\It$ into three bins with loads exactly $a,b,c$,
    \item that there exists a packing of items from $\It$ into three bins that does not exceed $T$ in any bin.
  \end{compactitem}  
\end{definition}

It is clear that every bin configuration is a valid state of the game
with \adversary as the next player. We may also observe that the existence of an online
algorithm for \binstretch implies an existence of an oblivious algorithm with the same
stretching factor that has access only to the current bin configuration $B$ and the incoming item $i$.

Using the concept of bin configuration and the previous two facts, we
may formally define the game we investigate:

\begin{definition}
For a given $S \in  \Nat, T \in  \Nat$, the \textbf{bin stretching game} $\Game(S,T)$ is the following two
player game:

\begin{itemize} 
  \item There are two players named \adversary and \algo. The player \adversary starts.
  \item Each turn of the player \adversary is associated with a bin configuration $B = (a,b,c,\It)$.
    The start of the game is associated with the bin configuration $(0,0,0,\emptyset)$.
  \item The player \adversary receives a bin configuration $B$. Then,
    \adversary selects a number $i$ such that the multiset $\It \cup
    \{i\}$ can be packed by an offline optimum into three bins of
    capacity $T$. The pair $(B,i)$ is then sent to the player \algo.

\item The player \algo receives a pair $(B,i)$. The player \algo has
to pack the item $i$ into the three bins as described in $B$ so that
each bin has load strictly less than $S$. \algo then updates the
configuration $B$ into a new bin configuration, denoted $B'$. \algo
then sends $B'$ to the player \adversary.
\end{itemize}
For a bin configuration $B$ we define recursively whether it is won or lost for player \adversary:
\begin{itemize}
  \item If the player \algo receives a pair $(B,i)$ such that it cannot pack
    the item according to the rules, the bin configuration $B$ is won for player \adversary.
    
  \item If the player \adversary has no more items $i$ that it can send from a
    configuration $B$, the bin configuration $B$ is lost for player \adversary.
    
  \item For any bin configuration $B$ where the player \adversary has
    a possible move, the configuration is won for player
    \adversary if and only if the game ends in a bin configuration $C$
    that is won for the player \adversary no matter which decision is
    made by the player \algo at any point.
\end{itemize}
\end{definition}

\begin{definition}
We say that a game $\Game(S,T)$ is a \textbf{lower bound} if and only
if the bin configuration $(0,0,0,\emptyset)$ is won for the player \adversary.
\end{definition}

\subsection{The minimax algorithm}\label{sec:minimax}

Our implemented algorithm is a fairly standard implementation of the
minimax game search algorithm.  The pecularities of our algorithm
(caching, pruning, and other details) are described in the following sections.

One of the differences between our algorithm and the algorithm of
Gabay et al. \cite{gabay2013lbv2} is that our algorithm makes no use of
alpha-beta pruning -- indeed, as every bin configuration is either won
for \algo or won for \adversary, there is no need to use this type of
pruning.

The following procedures return 0 if the bin configuration is won for
the player \adversary; otherwise they return 1 (player \algo wins).

\algobox{{\bf Procedure} $\evaladv$:

\noindent Input is a bin configuration $B = (a,b,c,\It)$.
\begin{compactenum}[(1)]
\item Check if the bin configuration is cached (Section \ref{subsec:caching}); if so, output the value found in cache and return.
% \item\label{step:heuristics} Apply \adversary-based heuristics (Section \ref{subsec:heuristics}).
\item Create a list $L$ of items which can be sent as the next step of the player \adversary (Section \ref{subsec:test}).
\item For every item size $i$ in the list $L$:
\item\indentskip Recurse by running $\evalalg(B,i)$.
\item\indentskip If $\evalalg(B,i)$ returns $0$, stop the cycle, store the configuration in the cache and end $\evaladv$ with value $0$.
\item\indentskip Otherwise, continue with the next item size.
\item If the evaluation reaches this step, store the configuration in the cache and return value $1$.
\end{compactenum}
}

\algobox{{\bf Procedure} $\evalalg$:

\noindent Input is a bin configuration $B = (a,b,c,\It)$ and item $i$.
\begin{compactenum}[(1)]
\item If applicable, prune the tree using known algorithms (Section \ref{subsec:gs}).
\item For any one of the three bins:
\item\indentskip If $i$ can be packed into the bin so that its load is less than $T$:
\item\indentskip\indentskip Create a configuration $B'$ that corresponds to this packing.
\item\indentskip\indentskip Run $\evaladv(B')$. If $\evaladv(B')$ returns 1, exit the procedure with value 1 as well.
\item\indentskip\indentskip Otherwise, continue with another bin.
\item If we reach this step, no placement of $i$ results in victory of \algo. We return 0 and exit.
\end{compactenum}
}

\algobox{{\bf Procedure} $\Main$:

\noindent Input is a bin configuration $B = (a,b,c,\It)$. 
\begin{compactenum}[(1)]
\item Fix parameters $S,T$.
\item Run $\evaladv(B)$.
\item If $\evaladv(B)$ returns $1$ (the game is won for player \algo), report failure.
\item Otherwise report success and output the game tree.
\end{compactenum}
}

\subsection{Verifying the offline optimum guarantee}\label{subsec:test}

When we evaluate a turn of the \adversary, we need to create the list
$L = \{0,1,\ldots,y\} ⊆ \{0,1,\ldots,T\}$ of items that \adversary can actually
send while satifying the \binstretch guarantee. We employ the
following steps:

\begin{enumerate}
\item First, we calculate a lower and upper bound $LB \leq UB$ on the maximal value $y$ of $L$.
\item Then, we do a linear search on the interval $\{UB,UB-1,\ldots,LB\}$ using a
procedure $\Test$ that checks a single multiset $\It'$, where $\It'$ is $\It$ plus the item in question and $\It$ is the current multiset of items.
\item The first feasible item size is the desired value of $y$.
\end{enumerate}

Note that in the second step we could also implement a binary search over the interval,
but in our experiments the difference between $UB$ and $LB$ was very small (usually at most 4),
thus a linear search is quicker.

\noindent {\bf Upper and lower bounds.} The running time of procedure
$\Test$ will be cubic in terms of $T$ in the worst case. We therefore
reduce the number of calls to $\Test$ by creating good lower and upper
bounds on the maximal item $y$ which \adversary can send.

To find a good lower bound, we employ a standard bin packing algorithm
called \bfd. \bfd packs items from $\It$ into three bins of capacity
$T$ with items in decreasing order, packing an item into a bin where
it ``fits best'' -- where it minimizes the empty space of a bin. \bfd
is a linear-time algorithm (it does not need to sort items in $\It$,
as the implementation of $\It$ stores them in a sorted order).

Our desired lower bound $LB$ will be the maximum empty space over all
three bins, after \bfd has ended packing. Such an item can always be
sent without invalidating the \binstretch guarantee.

Our upper bound $UB$ is comparatively simpler; for a bin configuration
$(a,b,c,\It)$, it will be set to $\min(T,3T - a - b -c)$. Clearly, no larger
item can be sent without raising the total size of all items above
$3T$.

\noindent {\bf Procedure $\Test$.} Procedure $\Test$ is a sparse
modification of the standard dynamic programming algorithm for
\textsc{Knapsack}. Given a multiset $\It, |\It| = n$, on input, our
task is to check whether it can be packed into three bins (knapsacks)
of capacity $T$ each.

We use a queue-based algorithm that generates a queue $Q_i$ of all
valid triples $(a,b,c)$ that can arise by packing the first $i$ items.

To generate a queue $Q_{i+1}$, we traverse the old queue $Q_i$ and add
the new item $\It[i+1]$ to the first, second and third bin, creating
up to three triples that need to be added to $Q_{i+1}$.

We make sure that we do not add a triple several times during one
step, we mark its addition into a auxilliary $\{0,1\}$ array $F$. Note
that the queue $Q_{i+1}$ needs only $Q_{i}$ and the item $\It[i+1]$ for
its construction, and so we can save space by switching between queues
$Q_1$ and $Q_2$, where $Q_{2i+1} = Q_{1}$ and $Q_{2i} = Q_{2}$.

The time complexity of the procedure $\Test$ is $\O(|\It|\cdot T^3)$
in the worst case.  However, when a bin configuration contains
large items, the size of the queue is substantially limited and the
actual running time is much better.

\algobox{\textbf{Procedure} $\Test$:

\noindent Input is a multiset of items $\It$.
\begin{compactenum}[(1)]
\item Create two queues $Q_1, Q_2$.
\item Add the triple $(\It[1],0,0)$ to $Q_1$.
\item For each item $i$ in the multiset $\It$, starting with the second item:
\item\indentskip For each triple $(a,b,c) \in  Q_1$:
\item\indentskip\indentskip If $a+s(i)\le T$:
\item\indentskip\indentskip\indentskip Add the triple $(a+s(i),b,c)$ to $Q_2$ unless $F[a+s(i),b,c] = 1$.
\item\indentskip\indentskip\indentskip Set $F[a+s(i),b,c] = 1$.
\item\indentskip\indentskip Do the same for triples $(a,b+s(i),c)$ and $(a,b,c+s(i))$.
\item\indentskip Swap the queues $Q_1$ and $Q_2$.
% \item Erase the array $F$.
\item Return True if the queue $Q_1$ is non-empty, False otherwise.
\end{compactenum}
}

{\bf Notes:} We employ two small optimizations that were not yet
mentioned. First, we sort the numbers $(a,b,c)$ in each triple to
ensure $a \ge b \ge c$, saving a small amount of space and time.
Second, we use one global array $F$ in order to avoid initializing it
with every call of the procedure $\Test$.

It is also worth noting that we could alternatively implement the
procedure $\Test$ using integer linear programming or using a CSP
solver (which has been done in \cite{gabay2013lbv2}). However, we
believe our sparse dynamic programming solution carries little
overhead and for large instances it is much faster than the CSP/ILP
solvers.

\subsection{Caching}\label{subsec:caching}

Our minimax algorithm employs extensive use of caching. We cache any
solved instance of procedure $\Test$ as well as any evaluated bin
configuration $B$ with its value. Note that we do not cache results of
Procedure $\evalalg$.

\noindent\textbf{Hash table limitation.} We store a large hash table of fixed size,
with each entry being a separate chain. With each node in a chain we store the number
of accesses. When a chain is to be filled over a fixed limit, we eliminate
a node with the least number of accesses.

To allow hash tables of variable size, our hash function returns a
$64$-bit number, which we trim to the desired size of our hash table.

In our definition of a bin configuration $(a,b,c,\It)$, we do not
require the loads $a,b,c$ to be sorted. However, configurations which
differ only by a permutation of the values $a,b,c$ are equivalent, and
so we sort these numbers when inserting a bin configuration into the
hash table.

\noindent\textbf{Hash function.} Our hash function is based on Zobrist hashing \cite{zobrist}, which
we now describe.

For each bin configuration, we count occurences of items, creating
pairs $(i,f) \in \{1,\ldots,T\} \times \{0,1,\ldots,3T\}$, where $i$ is the item type
and $f$ its frequency.  As an example, a bin configuration
$(3,2,3,\{1,1,1,1,2,3\})$ forms pairs $(1,4), (2,1), (3,1), (4,0),
(5,0)$ and so on.

At the start of our program, we associate a random $64$-bit number with each
pair $(i,f)$. We also associate a $64$-bit number for each possible load
of bin $A$, bin $B$ and bin $C$.

The Zobrist hash function is then simply a XOR of all associated
numbers for a particular bin configuration.

The main advantage of this approach is fast computation of new hash
values.  Suppose that we have a bin configuration $B$ with hash
$H$. After one round of the player \adversary and one round of the
player \algo, a new bin configuration $B'$ is formed, with one new
item placed. Calculating the hash $H'$ of $B'$ can be done in time
$\O(1)$, provided we remember the hash $H$ -- the new hash is
calculated by applying XOR to $H$, the new associated values, and the
previous associated values which have changed.

\noindent\textbf{Caching of the procedure $\Test$}. So far, we have described
caching of the bin configurations.  We also use the same approach for
caching the values of the procedure $\Test$.  To see the usefulness,
note that the procedure $\Test$ does not use the entire bin
configuration $B=(a,b,c,\It)$ as input, but only the multiset
$\It$. Therefore, we aim to eliminate overhead that is caused by
calling $\Test$ on a different bin configuration, but with the same
multiset $\It$.

Our hash function and hash table approaches are the same in both
cases.

\subsection{Tree pruning}\label{subsec:gs}

Alongside the extensive caching described in Subsection
\ref{subsec:caching}, we also prune some bin configurations where it is
possible to prove that a simple online algorithm is able to finalize
the packing. Such a bin configuration is then clearly won for player
\algo, as it can follow the output of the online algorithm.

Such situation are called \emph{good situations}, same as in
Section~\ref{sec:gs}. We will make use of the first five good
situations from Section \ref{sec:gs}.
% PV why not use the other GS? If we have a reason, we can write it here.
% MB I do not see the reason now, so I will leave it unchanged.

Recall that in the bin stretching game $\Game(S,T)$, the player \algo
is trying to pack all three bins with capacity strictly below $S$,
which we can think of as capacity $S-1$.  Therefore, we set $S' =
S-1$ and use $S'$ in our definitions.

We restate the good situations \gs1 to \gs5 for an instance of
$\Game(S',T)$ for general $S',T$ with $\alpha = S' - T$ satisfying
$\alpha \ge T/3$, while in Section \ref{sec:gs} we formulate the
good situations only for $\Game(22,18)$. The proofs are however
equivalent and we omit them.

\setcounter{goodsit}{0}
\begin{gengoodsit}\label{lem:newgs1}
Given a bin configuration $(a,b,c,\It)$ such that $a + b \geq 2T - \alpha$ and $c$ is
arbitrary,
there exists an online algorithm that packs all remaining items into
three bins of capacity $S'$.\qed
\end{gengoodsit}

\begin{gengoodsit}\label{lem:newgs2}
Given a bin configuration $(a,b,c,\It)$ such that $a \in [T-2\alpha, \alpha]$ and $b$ and
$c$ are arbitrary, there exists
an online algorithm that packs all remaining items into three bins
of capacity $S'$.\qed
\end{gengoodsit}

\begin{gengoodsit}\label{lem:newgs3}
Given a bin configuration $(a,b,c,\It)$ such that $a \in [\frac{3}{2}(T-\alpha),S']$ and either 
{\rm(i)} $c \ge \alpha$ and $b$ is arbitrary or {\rm(ii)} $b+c\ge S'$, 
there exists an online algorithm that packs all remaining items into
three bins of capacity $S'$.\qed 
\end{gengoodsit}

\begin{gengoodsit}\label{lem:newgs4}
Given a bin configuration $(a,b,c,\It)$ such that 
$a + b \ge \frac{3}{2}(T-\alpha) + c/2$, $b < T-2\alpha$, and $c < T-2\alpha$,
there exists an online algorithm that packs all remaining items into
three bins of capacity $S'$.\qed
\end{gengoodsit}

\begin{gengoodsit}\label{lem:newgs5}
Suppose that we are given a bin configuration $(a,b,c,\It)$ such that
an item $i$ with $s(i)>\alpha$ is present in the multiset $\It$ and the
following holds: $a \ge  s(i), b \ge  (3T - 7\alpha)/2, b\le\alpha, c = 0$.  Then there exists an
algorithm that packs all remaining items into three bins of capacity
$S'$.\qed
\end{gengoodsit}

\subsection{Results}\label{sec:results}

Table \ref{table:results} summarizes our results. The paper of Gabay, Brauner
and Kotov \cite{gabay2013lbv2} contains results up to the denominator 20; we
include them in the table for completeness. Results after the denominator 20 are
new. Note that there may be a lower bound of size $56/41$ even though
none was found with this denominator; for instance, some lower bound
may reach $56/41$ using item sizes that are not multiples of
$1/41$.

\renewcommand{\arraystretch}{1.3}
\setlength{\tabcolsep}{5pt}
\begin{table}[t]
\begin{center}
\begin{tabular}{ c | @{\hskip 2em} l | @{\hskip 2em} l | @{\hskip 2em} l }\label{tab:results}
\textit{Target fraction} & \hskip -1.5em \textit{Decimal form} & \hskip -1.5em \textit{L. b. found} & \hskip -1.5em\textit{Elapsed time} \\
\hline
$19/14$ &  $1.3571$ & Yes & 2s. \\
$22/16$ & $1.375$ & No & 2s. \\
$26/19$ & $1.3684$ & No & 3s. \\
\hline
$30/22$ & $1.\overline{36}$ & No & 6s. \\
$33/24$ & $1.375$ & No & 5s. \\
$34/25$ & $1.36$ & \textbf{Yes} & 15s. \\
$37/27$ & $1.\overline{370}$ & No & 10s. \\
$41/30$ & $1.3\overline{6}$ & No & 32s.\\
$44/32$ & $1.375$ & No & 34s. \\
$45/33$ & $1.\overline{36}$ & \textbf{Yes} & 1min. 48s. \\
$48/35$ & $1.3714$ & No & 2min. 8s. \\
$52/38$ & $1.3684$ & No & 6min. 14s. \\
$55/40$ & $1.375$ & No & 3min. 6s. \\
$56/41$ & $1.3659$ & No &  30min. \\
\end{tabular}
\end{center}

\caption{Results produced by our minimax algorithm, along with elapsed
time. The column \textit{L. b. found} indicates whether a lower bound
was found when starting with the given granularity. Fractions lower
than $19/14$ and higher than $11/8$ are omitted. Results were computed
on a server with an AMD Opteron 6134 CPU and 64496 MB RAM. The size of
the hash table was set to $2^{25}$ with chain length $4$. In order to
normalize the speed of the program, the algorithm only checked for a
lower bound and did not generate the entire tree in the \textbf{Yes}
cases.}\label{table:results}
\end{table}

\subsection{Lower bound for four and five bins}\label{sec:generalization}

The notion of bin configuration (Definition \ref{dfn:binconf}) as well
as most of the minimax algorithm can be straightforwardly generalized
for $m>3$. When generalizing the algorithm for larger $m$, one must
expect a slowdown, as the complexity of the sparse dynamic programming
from Section~\ref{subsec:test} is now $\O(|\It|\cdot T^m)$.

One notion that does not generalize very well are the good situations
of Section~\ref{subsec:gs}. For instance, the formula $a + b \ge (m-1)T -
\alpha$ in the statement of Good Situation \ref{lem:newgs1} will be much
less useful as $m$ grows. Some good situations, like Good Situation
\ref{lem:newgs2}, have no clear generalization for growing
$m$.

Therefore, we disable the pruning using good situations whenever
computing a lower bound for $m>3$.

Despite a significant increase in time complexity, we were able to
produce results for $m=4$ and $m=5$. See Table \ref{table:mfour} for
our results on four and five bins.

\begin{table}[t]
\begin{center}
\begin{tabular}{ c | @{\hskip 2em} l | @{\hskip 2em} l | @{\hskip 2em} l | @{\hskip 2em} l }
\textit{Number of bins} & \hskip -1.5em \textit{Target fraction} & \hskip -1.5em \textit{Decimal form} & \hskip -1.5em \textit{L. b. found} & \hskip -1.5em\textit{Elapsed time} \\
\hline
$4$ bins & $19/14$ &  $1.3571$ & \textbf{Yes} & 18s. \\
$5$ bins & $19/14$ &  $1.3571$ & \textbf{Yes} & 25min. \\
\end{tabular}
\end{center}
\caption{Results produced by our minimax algorithm in the case of $4$ and $5$ bins.
Tested on the same machine and with the same parameters as in Table~\ref{table:results}.
}\label{table:mfour}
\end{table}

\subsection{Verification of the results}\label{sec:verification}

We give a compact representation of our game tree for the lower bound
of $45/33$ for $m=3$, which can be found in Appendix
\ref{sec:appendix}. The fully expanded representation, as given by our
algorithm, is a tree on 11053 vertices.

%TODO: exact number of vertices
For our lower bounds of $19/14$ for $m=4$ and $m=5$, the sheer size of
the tree (e.g. 4665 vertices for $m = 5$) prevents us from
presenting the game tree in its entirety. We therefore include the lower
bound along with the implementations, publishing it online at
\url{http://github.com/bohm/binstretch/}.

We have implemented a simple independent C++ program which verifies
that a given game tree is valid and accurate. While verifying our
lower bound manually may be laborious, verifying the correctness of
the C++ program should be manageable. The verifier is available along
with the rest of the programs and data.

% --- END SNIP ---
\iffalse 
\section{Conclusions}

With our algorithm for $m=3$, the remaining gap is
small. For arbitrary $m$, we have seen at the beginning of
Section~\ref{sec:bigm} that a significantly new approach would be
needed for an algorithm with a better stretching factor than $1.5$.
Thus, after the previous incremental results, our algorithm is the
final step of this line of study. It is quite surprising that there
are no lower bounds for $m>3$ larger than the easy bound of $4/3$.
\smallskip

\noindent{\bf Acknowledgment.} The authors thank Emese Bittner for
useful discussions during her visit to Charles University.
\fi

\newpage
\appendix
\section{Appendix: Lower bound of 45/33}\label{sec:appendix}
\begin{figure}
  \centering
  \includegraphics[scale=0.8]{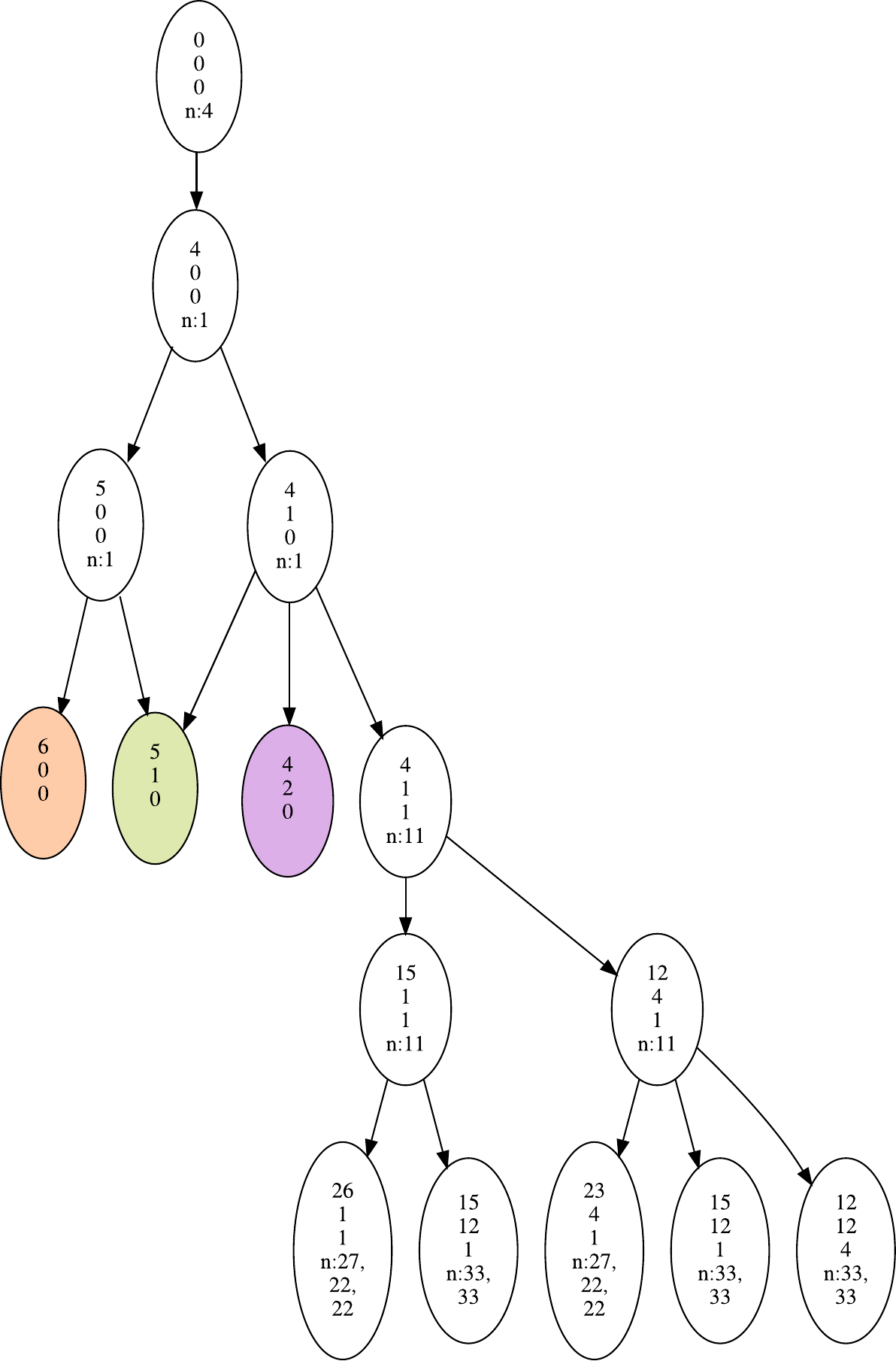}
  \caption{The beginning moves of the $45/33$ lower bound, scaled so
      that $T = 33$ and $S = 45$. The vertices contain the current
      loads of all three bins, and a string \texttt{n: $i$} with $i$
      being the next item presented by the \adversary. If there are
      several numbers after \texttt{n:}, the items are presented in
      the given order, regardless of packing by the player \algo. The coloured vertices are expanded in later figures.}
\end{figure}
\newpage
\newgeometry{margin=0.3in}
\begin{figure}
\centering
\includegraphics[scale=0.9]{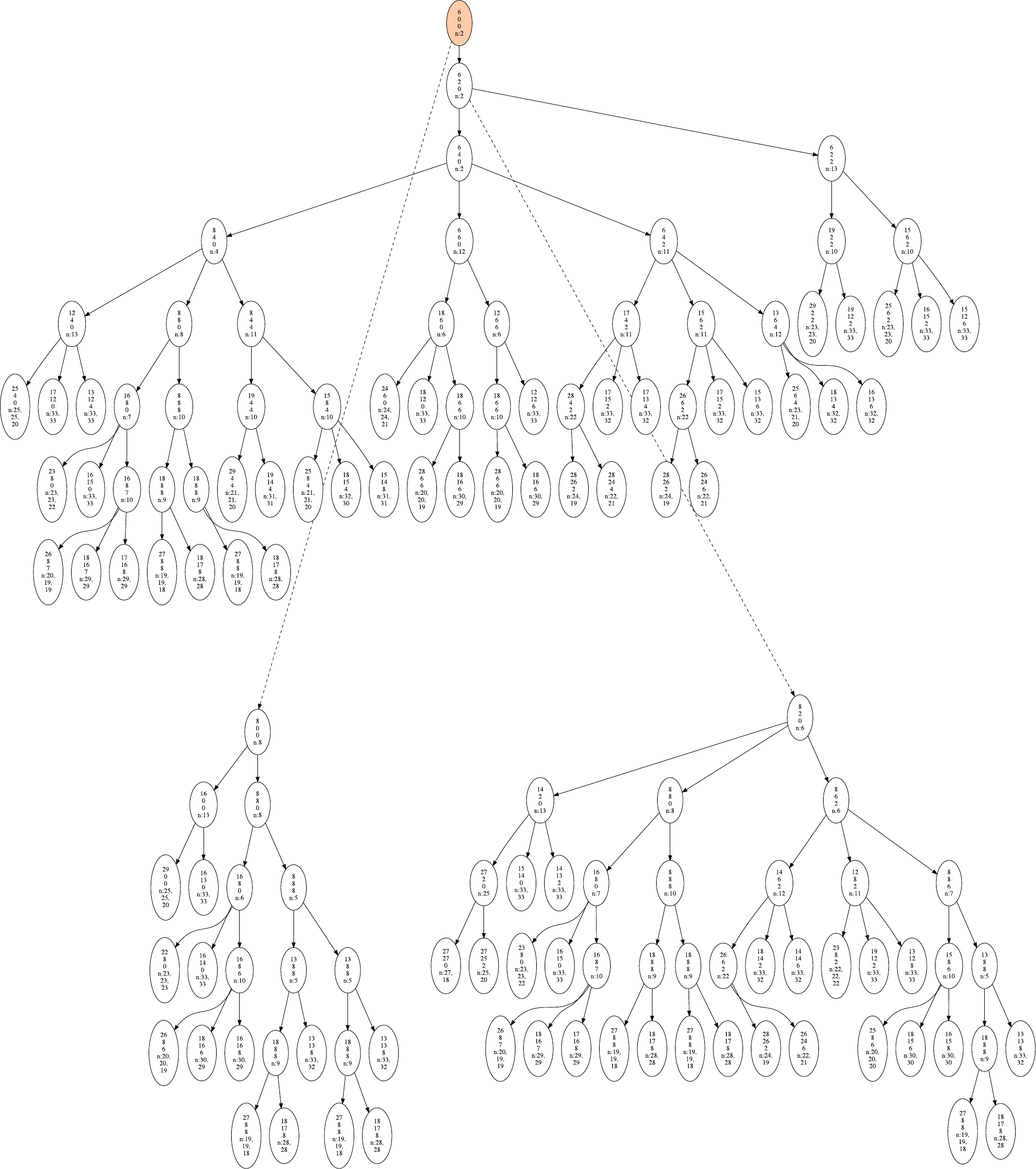}
\caption{Game tree for the lower bound of $45/33$, starting with the bin configuration $(6,0,0,\{4,1,1\})$.}
\end{figure}

\begin{figure}
\centering
\includegraphics[scale=0.9]{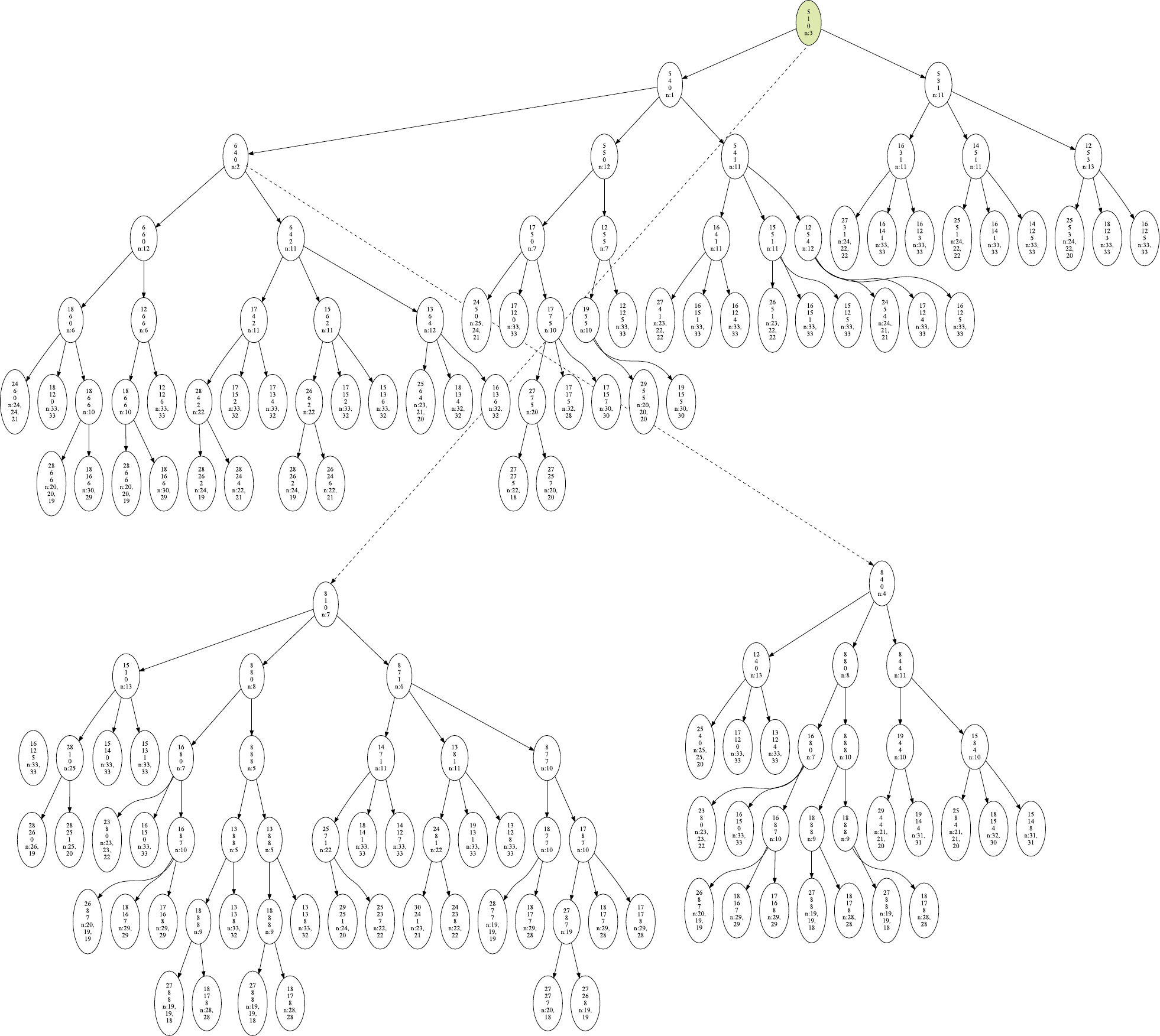}
\caption{Game tree for the lower bound of $45/33$, starting with the bin configuration $(5,1,0,\{4,1,1\})$.}
\end{figure}

\begin{figure}
\centering
  \includegraphics[scale=0.9]{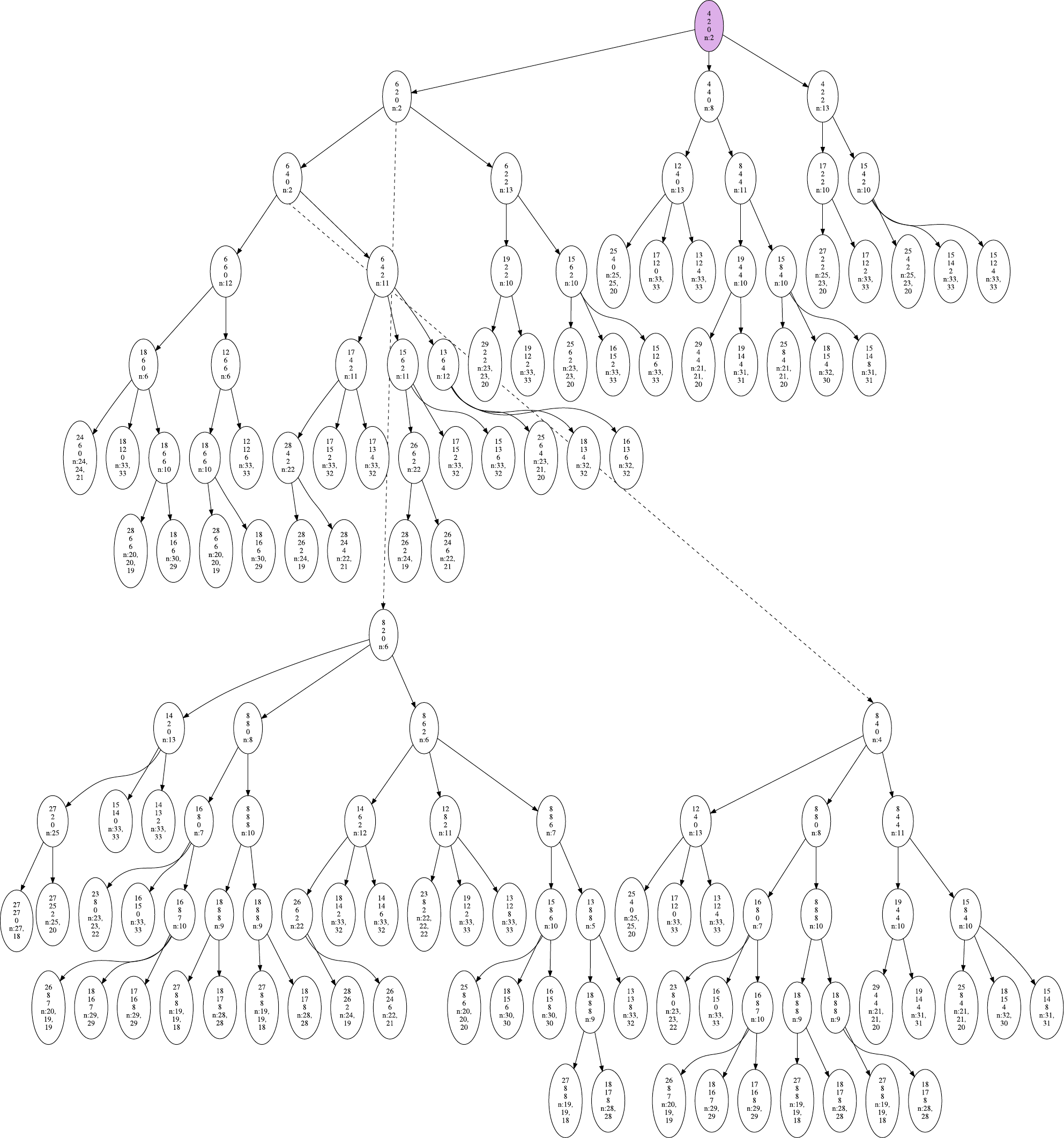}
  \caption{Game tree for the lower bound of $45/33$, starting with the bin configuration $(4,2,0,\{4,1,1\})$.}
\end{figure}
\restoregeometry

\end{document}